\documentclass[twocolumn]{autart}    

\usepackage{amsmath}
\usepackage{algorithm,algpseudocode}
\usepackage{amssymb}
\usepackage{color}
\usepackage{enumerate}
\usepackage{wrapfig}
\usepackage{graphicx}          
\usepackage[round, sort]{natbib}
\usepackage{url}

\newtheorem{theorem}{Theorem}[section]
\newtheorem{corollary}{Corollary}[section]
\newtheorem{remark}{Remark}[section]
\newtheorem{definition}{Definition}[section]

\newtheorem{assumption}{Assumption}[section]
\newtheorem{problem}{Problem}

\usepackage{hyperref}
\usepackage{comment}
\hypersetup{
	colorlinks = true,
	citecolor = cyan,
}

\makeatletter
\def\ps@copyright{%
  \let\@mkboth\@gobbletwo
  \def\@oddhead{}%
  \let\@evenhead\@oddhead
  \def\@oddfoot{}%
  \let\@evenfoot\@oddfoot
}
\makeatother
	
\begin{document}
	\begin{frontmatter}
			
	\title{Robust Data-Driven Nash Equilibrium Seeking under Partial-Decision Information} 
			
	\thanks[footnoteinfo]{The work was supported in part by the National Natural Science Foundation of China under Grants U23B2059 and 62088101. \emph{(Corresponding author: Gang Wang.)}
            }
			
		\author[Bit]{Linqi Wang}\ead{wanglinqi@bit.edu.cn},
        \author[Ntu]{Yifei Li}\ead{li.yifei@ntu.edu.sg}, 
        \author[Ntu]{Wenjie Liu}\ead{wenjie.liu@ntu.edu.sg}, 
        \author[Bit]{Yuzhou Wei}\ead{weiyuzhou@bit.edu.cn},
		\author[Bit]{Gang Wang}\ead{gangwang@bit.edu.cn}, 
        \author[Ntu]{Lihua Xie}\ead{elhxie@ntu.edu.sg} 
		
		\address[Bit]{State Key Lab of Autonomous Intelligent Unmanned Systems, Beijing Institute of Technology, Beijing 100081, China}  
    		\address[Ntu]{Centre for Advanced Robotics Technology  Innovation (CARTIN), School of Electrical and Electronic Engineering,\\ Nanyang Technological University, Singapore} 

		\maketitle


    \allowdisplaybreaks
	
	\begin{abstract}

    This paper presents a data-driven framework for decentralized Nash equilibrium (NE) seeking in multi-agent systems with unknown linear dynamics subject to exogenous disturbances, operating under partial-decision information (where agents lack direct access to the decisions of all others) and equality constraints. The proposed framework integrates an NE model, a distributed communication protocol, an internal model for disturbance rejection, and a data-driven stabilization strategy. By reformulating the problem as a cooperative output regulation problem, we synthesize controllers directly from noisy input-state data via semi-definite programs (SDPs), providing formal guarantees for closed-loop stability and asymptotic convergence to the NE. The approach is further extended to a class of nonlinear systems with constant disturbances by leveraging integral control and describing nonlinearities via quadratic constraints. Numerical simulations involving unmanned aerial vehicle  networks and a rotary-wing aerial vehicle formation validate the efficacy and robustness of the proposed method.

 \end{abstract}
		\end{frontmatter}	
\allowdisplaybreaks

\section{Introduction}\label{sec:intro}

Game theory provides a principled framework for analyzing strategic interactions in multi-agent dynamical systems, with the Nash equilibrium (NE) serving as a central solution concept under which no agent can unilaterally reduce its own cost \citep{bacsar1998dynamic}. While early studies mainly focused on static games \citep{frihauf2011nash}, many contemporary engineering applications, such as robotic coordination in sensor networks \citep{stankovic2011distributed}, distributed resource allocation \citep{Zhou2024Distributed,WANG2024111356}, and optimal power flow in smart microgrids \citep{scarabaggio2024local}, involve agents with complex dynamics, couplings, and operational constraints. These applications call for NE seeking mechanisms that are not only distributed, but also amenable to feedback implementation and robust to uncertainty and disturbances.

Considerable progress has been made on NE seeking for dynamical agents, including single- and double-integrator systems \citep{YE2025112074,WangRobus2022}, as well as more general linear \citep{he2025stochastic,LIU2026112603} and nonlinear systems \citep{xu2024prescribed}. However, much of the existing literature assumes that each agent has full access to the decisions of the other agents and does not explicitly account for exogenous disturbances. This limits its applicability in practice, since the decisions of other agents are often unavailable in networked environments, while rejecting external disturbances is an integral part of closed-loop design. These limitations naturally motivate the study of NE seeking under incomplete information and disturbance-affected dynamics.

A further challenge arises from \emph{partial-decision information}, where each agent has no direct access to the decisions of all other agents and must instead rely on local communication \citep{salehisadaghiani2018distributed,qian2021distributed,MENG2023110919}. This setting has been studied using distributed optimization and consensus-based techniques \citep{tatarenko2020geometric,huang2024distributed}. In particular, \citet{guo2021linear} reformulated NE seeking for linear systems as a distributed output regulation problem, and \citet{romano2025game} further integrated an NE model, a communication protocol, and a stabilizing controller into a unified architecture. However, these approaches are mainly developed for linear time-invariant agents.

The difficulty becomes even more pronounced for nonlinear multi-agent systems. For example, \citet{SASSANO2022110389} characterized open-loop NEs for nonlinear differential games, while \citet{zhang2019distributed} and \citet{huang2024distributed} studied distributed schemes for nonlinear agents with uncertainties. Compared with the linear case, nonlinear dynamics complicate both the representation of the closed-loop system and the derivation of tractable stability conditions. As a result, existing approaches typically rely on substantial prior structural knowledge of the nonlinearities, such as known nominal models or prescribed basis-function parameterizations. This strong model dependence significantly restricts their applicability when the system dynamics are unknown.

In this context, data-driven control offers an attractive alternative when model identification is costly or unreliable. Building on Willems' fundamental lemma \citep{willems2005note}, data-driven methods enable controller synthesis directly from measured data without an explicit identification step \citep{liu2022data,wang2026data,liu2023data}. 
This paradigm has been extended to a broad range of control problems, including distributed control \citep{11533487},  model predictive control \citep{WEI2026113010}, aperiodic control \citep{yuan2026data}, and event-triggered control \citep{zhang2026data}.
Recent studies have also applied data-driven techniques to dynamic games, ranging from data-driven feedback NE computation for linear quadratic games \citep{xu2025data} to output regulation-based NE seeking from data \citep{wang2026robust}. 
 Despite these advances, several fundamental challenges remain: (i) extending output regulation-based NE seeking to partial-decision information settings with equality constraints; (ii) synthesizing decentralized controllers for unknown linear agents directly from noisy data; and (iii) extending direct data-driven NE seeking to nonlinear agents while retaining tractable stability and convergence guarantees.

To address these challenges, this paper develops a data-driven NE seeking framework for multi-agent systems with unknown dynamics, partial-decision
information, equality constraints, and disturbances. Specifically, we recast the constrained NE seeking problem as a cooperative output regulation problem by integrating a KKT-based NE model for generating regulation errors, a communication protocol for recovering unavailable decision information, an internal model for disturbance rejection, and decentralized controllers for closed-loop stabilization into a unified architecture. For unknown linear agents subject to external disturbances generated by known exosystems, we formulate a robust data-driven semidefinite program (SDP) using noisy
input-state data, whose feasibility guarantees closed-loop stability and asymptotic convergence to the constrained NE. We further extend the framework
to a class of nonlinear agents with constant disturbances by combining integral control with quadratic constraints on the nonlinearities. The
effectiveness and applicability of the proposed approach are demonstrated through numerical simulations involving unmanned aerial vehicle networks and rotary-wing aerial vehicle formation.

The main contributions are summarized as follows.
\begin{itemize}
  \item  A unified control framework is developed by integrating an NE model, a communication protocol, an internal model, and a data-driven controller, reformulating the NE seeking problem as a cooperative output regulation problem.

  \item For unknown linear systems subject to exogenous disturbances generated by known exosystems, we propose a data-driven design method based on SDPs, which guarantees closed-loop stability and asymptotic convergence to the NE using  noisy input-state data.

  \item An extension to a class of nonlinear systems with constant disturbances is developed by leveraging integral control and quadratic constraints, which enables decentralized data-driven NE seeking with asymptotic convergence guarantees.
\end{itemize}

It is worth clarifying the relationship between this work and the most relevant literature \citep{wang2026robust,guo2021linear,romano2025game,huang2024distributed}. Although the analysis uses several standard tools, the present results are not obtained by a straightforward combination of existing methods.

\begin{itemize}
\item \textbf{Compared with our previous work \citep{wang2026robust},} which considers unconstrained games in a standard network game setting with aligned cost and communication structures, the present paper studies a more practically relevant setting with equality constraints and partial-decision information. This requires the controller design to explicitly incorporate KKT dynamics and distributed communication protocols into the augmented system, thereby changing the structure of the decentralized synthesis problem and the associated nonlinear analysis.
\item \textbf{Compared with the model-based approaches in \citep{guo2021linear,romano2025game,huang2024distributed},} the present paper develops a unified direct data-driven framework without relying on exact system matrices. In particular, unlike \citet{guo2021linear}, we consider partial-decision information and do not require the cost-coupling structure to coincide with the communication topology. Unlike \citet{romano2025game}, our linear framework accommodates external disturbances generated by dynamic exosystems and is further complemented by a nonlinear extension for constant disturbances. Unlike \citet{huang2024distributed}, our framework covers equality-constrained games under partial-decision information and a broad class of nonlinear systems admitting lifted representations and quadratic constraints.
\end{itemize}

\subsection{Notation}
Let $\mathbb{R}$  denote the set of real numbers. For a matrix $Y$, $Y \succ (\succeq) 0$ implies that $Y$ is positive (semi-)definite, and 
$Y \prec (\preceq) 0$ implies $Y$ is negative  (semi-)definite. The spectral norm of $Y$ is denoted by $\|Y\|$, and the Moore-Penrose pseudoinverse of $Y$ is denoted by $Y^\dagger$. The Euclidean norm of a vector $x \in \mathbb{R}^n$ is denoted by $\|x\|$.  Given vectors $x_{1},\dots,x_{n}$, the stacking operator is defined as $\operatorname{col}(x_{1},\dots,x_{n})=[\,x_{1}^{\top},\dots,x_{n}^{\top}\,]^{\top}$. The operator $\operatorname{blkdiag}\{\cdot\}$ denotes a block diagonal matrix. The symbol $I_{n}$ denotes the $n$-dimensional identity matrix. We abbreviate the symmetric matrix $\begin{bmatrix}
\begin{smallmatrix}
A&B\\
B^{\top}&C
\end{smallmatrix}\end{bmatrix}$ as $\begin{bmatrix}\begin{smallmatrix}
A&B\\
\star &C
\end{smallmatrix}\end{bmatrix}$. 
For a finite set $\mathcal{X}$, its cardinality is denoted by 
$|\mathcal{X}|$. 

\section{Problem Formulation}\label{sec_pre}

This section formalizes the NE seeking problem for a class of multi-agent systems with unknown linear dynamics and gives the definitions and preliminaries that form the basis for subsequent analysis.

Consider a game with $N$ agents indexed by the set $\mathcal{I}=\{1,\ldots,N\}$. The dynamics of agent $i\in\mathcal{I}$ are given by
\begin{equation}\label{linsys}
\begin{aligned}
\dot{x}_i &= A_ix_i + B_iu_i + D_iw_i, \\
y_i  &= C_ix_i,
\end{aligned}
\end{equation}
where $x_i \in \mathbb{R}^{n_{i}}$ is the state, $u_i \in \mathbb{R}^{m_{i}}$ the control input (strategy vector), $y_i \in \mathbb{R}^{p_{i}}$ the output (decision variable), and $w_i \in \mathbb{R}^{q_{i}}$ the disturbance. 
The system matrices $A_i \in \mathbb{R}^{n_i \times n_i}$, $B_i \in \mathbb{R}^{n_i \times m_i}$, $C_i \in \mathbb{R}^{p_i \times n_i}$, and $D_i \in \mathbb{R}^{n_i \times q_i}$ are unknown.
The disturbance $w_i$ is generated by an exosystem
\begin{equation}\label{disturbance}
\dot{w}_i = E_{i}w_i,
\end{equation}
where $E_i \in \mathbb{R}^{q_i \times q_i}$ satisfies the following assumption.
\begin{assumption}\label{no_eig}
For each $i\in \mathcal{I}$, $E_{i}$ is known and has no eigenvalues
with negative real parts.
\end{assumption}


We consider a partial-decision information setting, in which agent $i\in\mathcal I$ does not have direct access to all other agents' decisions. Accordingly, each agent maintains local estimates $\hat y_j^i$ of $y_j$ for all $j\in\mathcal I$, with $\hat y_i^i:=y_i$.
The communication network is modeled by an undirected graph $G=(\mathcal I,\mathcal E,\mathcal A)$, where $\mathcal E\subseteq\mathcal I\times\mathcal I$ is the edge set and $\mathcal A=[a_{ij}]\in\mathbb R^{N\times N}$ is the adjacency matrix. The adjacency weights satisfy $a_{ii}=0$, $a_{ij}=a_{ji}=1$ if $(i,j)\in\mathcal E$, and $a_{ij}=0$ otherwise. The neighbor set of agent $i$ is $\mathcal N_i:=\{j\in\mathcal I:(i,j)\in\mathcal E\}.$ To ensure effective information transmission, we impose the following assumption.

\begin{assumption}\label{ass_graph}
The communication graph $G$ is unweighted, undirected, and connected.
\end{assumption}

Building on these settings, each agent $i$ seeks to minimize a local quadratic cost function 
\begin{align}\label{cost_func}
J_i(y_i,y_{-i})
&= y_i^\top Q_{ii}y_i + F_{ii}y_i + \rho_i \nonumber\\
&\quad + \sum_{j \in \mathcal{I} \setminus \{i\}}\bigl(y_i^\top Q_{ij}y_j + y_j^\top F_{ij}y_j\bigr),
\end{align}
where $y_{-i}={\rm col}(y_1,\ldots,y_{i-1},y_{i+1},\ldots,y_N)$, $Q_{ii}\in\mathbb R^{p_i\times p_i}$ with $Q_{ii}\succ0$, $F_{ii}\in\mathbb R^{1\times p_i}$, $\rho_i\in\mathbb R$, $Q_{ij}\in\mathbb R^{p_i\times p_j}$, and $F_{ij}\in\mathbb R^{p_j\times p_j}$ for $i\neq j$.
In general, the cost couplings are not restricted to the communication graph. Therefore, under partial-decision information, agent $i$ may need estimates of non-neighboring decisions to evaluate its local gradient.
\begin{remark}
The cost function in \eqref{cost_func} is more general than the standard network game cost, since its coupling terms are not restricted to the communication graph. Indeed, the summation in  \eqref{cost_func} is over all  $j \in \mathcal{I} \setminus \{i\}$ and nonzero $Q_{ij}$ and $F_{ij}$ may appear even for $j \notin \mathcal{N}_i$. Hence, under partial-decision information, agent $i$ may need estimates of non-neighboring decisions for local gradient evaluation. The standard network game is recovered as a special case. In particular,
\begin{align}\label{cost_exp_2}
J_i=\|y_i-r_i\|^2+\sum_{j\in\mathcal N_i}\|y_i-y_j\|^2
\end{align}
is recovered from \eqref{cost_func} by choosing
$Q_{ii}=(1+|\mathcal N_i|)I_{p_i}$, $F_{ii}=-2r_i^\top$, $\rho_i=r_i^\top r_i$,
$Q_{ij}=-2I_{p_i}$ and $F_{ij}=I_{p_j}$ for $j\in\mathcal N_i$, and
$Q_{ij}=0$, $F_{ij}=0$ otherwise.
\end{remark}

In addition, for engineering considerations, we impose the following linear constraint
\begin{equation}\label{equa_cons}
R_i y_i = h_i,
\end{equation}
where  $R_{i}\in\mathbb{R}^{f_{i} \times p_{i}}$ and $h_{i}\in\mathbb{R}^{f_{i}}$. The feasible set for agent $i$ is given by $\Lambda_i  = \{ y_i \in \mathbb{R}^{p_i} \mid  R_i y_i = h_i \}$. The resulting optimization problem for agent $i$ can then be written as 
\begin{subequations}\label{min_cost}
\begin{align}
\min_{y_i} ~&~J_{i}(y_{i},y_{-i}), \label{min_cost_a}\\
\text{s.t.}~&~   R_i y_i = h_i.\label{min_cost_c}
\end{align}
\end{subequations}

\begin{remark}
The optimization problem in \eqref{min_cost} extends the formulation in \cite{guo2021linear} by explicitly incorporating the linear constraint \eqref{equa_cons}, which reflects engineering requirements. Such equality constraints frequently arise in practice, e.g., power balance in optimal power flow \citep{10742555}, and formation centering in multi-robot coordination \citep{chen2020sign}. 
This inclusion broadens the framework to scenarios where decision variables must satisfy strict balance conditions.
\end{remark}

The objective is to seek an NE of the game \eqref{min_cost}, whose formal definition is given below.

\begin{definition}
The output $y^*:={\rm col}(y_{1}^{*},\ldots,y_{N}^{*})$ is an NE of \eqref{min_cost} if, for every agent $i\in\mathcal{I}$,
\begin{equation}
J_i(y^*_i, y^*_{-i}) \leq J_i(y_i, y^*_{-i}), ~\forall y_i \in \Lambda_i.
\end{equation}
\end{definition}
Since $Q_{ii} \succ 0$, each local cost function $J_i(y_i,y_{-i})$ is continuous, strictly convex, and radially unbounded in $y_i$ for any fixed $y_{-i}$. It then follows from \cite[Corollary 4.2]{bacsar1998dynamic} that the game \eqref{min_cost} admits an NE.

Define the partial gradient of $J_i(y_i,y_{-i})$ with respect to $y_i$ as
$\nabla_i J_i(y_i,y_{-i}) := \frac{\partial J_i(y_i,y_{-i})}{\partial y_i}
= \operatorname{col}\!\left(\frac{\partial J_i}{\partial y_{i1}},\ldots,\frac{\partial J_i}{\partial y_{ip_i}}\right)\in\mathbb{R}^{p_i}$. 
Let the pseudo-gradient be
$f(y)=\operatorname{col}(\nabla_1J_1(y_1,y_{-1})$, $\ldots,\nabla_NJ_N(y_N,y_{-N}))$,
where $y=\operatorname{col}(y_1,\ldots,y_N)$. 
The pseudo-gradient admits the compact representation 
\begin{equation}\label{barQ}
f(y) = \bar{Q}y + \bar{F},
\end{equation} 
where $\bar F=\operatorname{col}(F_{11}^\top,\ldots,F_{NN}^\top)
$ and 
\begin{equation*}
\bar{Q} =
\begin{bmatrix}
Q_{11} + Q_{11}^\top & Q_{12} & \cdots & Q_{1N} \\
Q_{21} & Q_{22} + Q_{22}^\top & \cdots & Q_{2N} \\
\vdots & \vdots & \ddots & \vdots \\
Q_{N1} & Q_{N2} & \cdots & Q_{NN} + Q_{NN}^\top
\end{bmatrix}.
\end{equation*}  
For the equality-constrained game \eqref{min_cost}, by the KKT conditions for each agent's convex optimization problem \cite[Section 5.5.3]{boyd2004convex}, $y^*$ is a constrained NE if and only if there exists a multiplier vector
$\varsigma^*=\operatorname{col}(\varsigma_1^*,\ldots,\varsigma_N^*)$, with $\varsigma_i^*\in\mathbb{R}^{f_i}$, such that
\begin{equation}\label{equ_Q}
\bar{Q}y^* + \bar{F} + R^\top \varsigma^* = 0, \quad Ry^* = h,
\end{equation}
where $R = \mathrm{blkdiag}(R_1,\dots,R_N)$ and $h = \mathrm{col}(h_1,\dots,h_N)$.

\begin{assumption}\label{ass_pseudo}
The feasible set $\Lambda = \prod_{i\in\mathcal{I}} \Lambda_i$ is nonempty, and the KKT system
\begin{align}\label{barq}
    \bar{Q}y + \bar{F} + R^\top \varsigma = 0, \quad Ry = h
\end{align}
admits a unique solution $(y^*, \varsigma^*)$.
\end{assumption}

\begin{remark}\label{rem_kkt_uni}
Assumption \ref{ass_pseudo} guarantees that the KKT system  \eqref{barq} has a unique solution and thus that the equality-constrained game admits a unique NE. All subsequent convergence results rely on this assumption. If Assumption \ref{ass_pseudo} is relaxed, the game may admit multiple NEs, and the present framework does not provide an a priori criterion for equilibrium selection. A rigorous treatment of the multiple-equilibrium case would require additional tools, such as explicit equilibrium-selection mechanisms or refined equilibrium concepts; see, e.g., \cite{benenati2023optimal,kulkarni2012variational}. These extensions are left for future work.
\end{remark}

\begin{remark}\label{KKT-matrix}
Assumption \ref{ass_pseudo} holds if and only if the KKT matrix
$K:=\begin{bmatrix}
\begin{smallmatrix}
\bar Q & R^\top\\
R & 0
\end{smallmatrix}\end{bmatrix}$ 
is nonsingular. If $\bar Q$ is nonsingular, then by the Schur complement, $K$ is nonsingular if and only if $R\bar Q^{-1}R^\top$ is nonsingular. 
Therefore, a convenient sufficient condition is that $\bar Q\succ 0$ and $R$ has full row rank, which imply
$R\bar Q^{-1}R^\top \succ 0$. Moreover, a  sufficient condition for $\bar Q$ to be nonsingular is block strict diagonal dominance, i.e.,
$\left\|(Q_{ii}+Q_{ii}^\top)^{-1}\right\|^{-1}
>
\sum_{j\in\mathcal I,\;j\neq i}\|Q_{ij}\|$, $\forall i\in\mathcal I$; see, e.g., \citep[Theorem~1]{feingold1962block}.
\end{remark}

Under Assumption \ref{ass_pseudo}, the equality-constrained game \eqref{min_cost} admits a unique NE $y^\ast$. The problem considered in this paper is to design decentralized controllers that steer the agent outputs to $y^\ast$ under unknown system dynamics, exogenous disturbances, equality constraints, and partial-decision information. This is formalized below.

\begin{problem}\label{pro_1}
Under Assumptions \ref{no_eig}--\ref{ass_pseudo}, consider the equality-constrained game \eqref{min_cost} for  $N$ agents governed by the unknown linear dynamics \eqref{linsys} and subject to disturbances $w_i$ generated by \eqref{disturbance}.
In the partial-decision information setting, the goal is to design a decentralized control strategy $u_{i}$ for each agent $i\in\mathcal{I}$ such that
(i) the closed-loop system admits an asymptotically stable equilibrium; and (ii) the output $y$ converges to the NE $y^*$.
\end{problem}

\section{The Control Framework}
To solve Problem \ref{pro_1}, this section generalizes the framework in \citet{romano2025game} by integrating four components: (i) an NE model for error approximation; (ii) a partial-decision information communication protocol for inter-agent information exchange; (iii) an internal model for disturbance rejection; and (iv) a decentralized data-driven controller for stabilizing unknown dynamics. 
Unlike \citet{romano2025game}, where integral control compensates only constant disturbances, the internal model rejects disturbances generated by exosystems. Moreover, to address the NE seeking problem without knowledge of the matrices $(A_{i}, B_{i}, C_{i}, D_{i})$, we introduce a novel data-driven controller. The design of each component is outlined below.

\subsection{The NE Model}
The constrained NE is characterized by the KKT system in \eqref{equ_Q}. Since the equilibrium output $y_i^\ast$ is not known a priori, the tracking error $y_i-y_i^\ast$ cannot be directly used for feedback design. Instead, each agent uses a dynamical system to generate a surrogate for this error, which relies on the computation of $\nabla_i J_i(y_i,y_{-i})$.
However, under partial-decision information, agent $i$ does not have direct access to all components of $y_{-i}$. Hence, each agent maintains local estimates of the decision profile. Specifically, define 
\begin{subequations}\label{eq:est_y} 
\begin{align}
    \hat{y}_i^i &:= y_i, \\\hat{y}_{-i}^i &:= \text{col}(\hat{y}_j^i)_{j \neq i}, \\\hat{y}^i &:= \text{col}(\hat{y}_j^i)_{j \in \mathcal{I}}. 
\end{align}
\end{subequations}
Leveraging these estimates and inspired by \citet{lawrence2020linear}, we introduce the following NE model
\begin{subequations}\label{eq:NEM} 
\begin{align}
\dot{\varsigma}_i &= R_i y_i - h_{i},\label{eq:NEMa}\\
e_i &= \nabla_i J_i(y_i, \hat{y}_{-i}^{i}) + R_i^\top \varsigma_i, \label{eq:NEMb}
\end{align}
\end{subequations}
where  $\varsigma_i\in\mathbb R^{f_i}$ is a
local multiplier state and $e_i\in\mathbb R^{p_i}$ is the local regulation error.
Therefore, the NE seeking problem can be cast into an output regulation problem, where driving $e_i$ to zero enforces stationarity, while driving $\dot{\varsigma}_i$ to zero enforces feasibility.  
Under Assumption \ref{ass_pseudo}, these conditions have a unique solution, and therefore the output is the constrained NE $y^\ast$.
\begin{remark}
Although equality constraints are standard in optimization, their incorporation is nontrivial in the present framework. First, they change the equilibrium characterization from a pure pseudo-gradient condition to a KKT system with multiplier dynamics. Second, under partial-decision information, the corresponding KKT conditions cannot be implemented directly from locally available information and therefore require communication protocols for their distributed realization. Third, the associated multiplier dynamics become part of the augmented closed-loop system used for offline data collection and decentralized data-driven synthesis, thereby substantially changing the structure of the subsequent controller design problem.
\end{remark}

\subsection{The Partial-Decision Information Communication Protocol}
Under partial-decision information, agent $i$ does not have direct access to all components of $y_{-i}$. To enable gradient evaluation using only locally available information, we employ a consensus-based communication protocol to generate local estimates of the decision profile.
Accordingly, the joint output is partitioned into three parts, namely the preceding part $p_{pi}:=\sum_{j<i}p_j$, the local part $p_i$, and the subsequent part $p_{si}:=\sum_{j>i}p_j$.
Then define the selection matrices
\begin{align*}
\mathbf{M}_{y_i}&:=
\begin{bmatrix}
0_{p_i \times p_{pi}} & I_{p_i} & 0_{p_i \times p_{si}}
\end{bmatrix},\\
\mathbf{H}_{y_i}&:=
\begin{bmatrix}
I_{p_{pi}} & 0_{p_{pi} \times p_i} & 0_{p_{pi} \times p_{si}} \\
0_{p_{si} \times p_{pi}} & 0_{p_{si} \times p_i} & I_{p_{si}}
\end{bmatrix},
\end{align*}
so that $y_i = \mathbf{M}_{y_i} \hat{y}^i$, $\hat{y}_{-i}^i = \mathbf{H}_{y_i} \hat{y}^i$, and $\hat{y}^i = \mathbf{M}_{y_i}^\top y_i + \mathbf{H}_{y_i}^\top \hat{y}_{-i}^{i} $.
Each agent updates its local estimate $\hat y_{-i}^i$ through the following consensus-based communication protocol over $G$
\begin{equation}
\dot{\hat{y}}_{-i}^i = \delta \sum_{j \in\mathcal{I}} a_{ij} (\mathbf{H}_{y_i} \hat{y}^j - \hat{y}_{-i}^i),\label{eq:cpy}
\end{equation}
where $\delta > 0$ is a connectivity parameter. 
Under the proposed protocol, the estimates asymptotically recover
$\hat{y}_{-i}^{i}\rightarrow y_{-i}$,
thereby enabling accurate gradient evaluation \cite[Definition 2]{romano2022game}.
\begin{remark}
Although the communication protocol in \eqref{eq:cpy} has a consensus-like form, it should not be interpreted as a standard consensus protocol. Standard consensus protocols act on the actual agent states or outputs and are designed to drive them to asymptotic agreement; see, e.g., \cite{olfati2007consensus,ren2005consensus}. By contrast, \eqref{eq:cpy} acts only on the estimate variables and is introduced to reconstruct the unavailable components of the decision profile required for local gradient evaluation under partial-decision information. Its role is therefore not to enforce agreement of the actual outputs $y_{i}$, but to support the distributed implementation of the NE model \eqref{eq:NEM}.
\end{remark}

\subsection{The Internal Model}
To reject the exogenous disturbances generated by \eqref{disturbance}, each agent $i \in \mathcal{I}$ is equipped with an internal model. 
            
\textbf{Definition 3.1} (Internal Model \cite{huang2004nonlinear}).
For the exosystem $\dot{w}_i = E_i w_i$, the matrix pair $(\Sigma_i, \hat{\Sigma}_i)$ forms a $p_i$-copy internal model of $E_i$ if
\begin{align*}
\Sigma_{i}=\underbrace{\operatorname{blkdiag}(\sigma_{i},\dots,\sigma_{i})}_{p_{i}\rm{-tuple}},~
\hat{\Sigma}_{i}=\underbrace{\operatorname{blkdiag}(\hat{\sigma}_{i},\dots,\hat{\sigma}_{i})}_{p_{i}\rm{-tuple}},
\end{align*}
where $\sigma_i$ is a constant matrix whose characteristic polynomial equals the minimal polynomial of $E_i$, and $\hat{\sigma}_i$ is a constant vector such that $(\sigma_i, \hat{\sigma}_i)$ is controllable.

Equipped with this matrix pair, each agent $i\in\mathcal{I}$ implements the internal model via
\begin{equation}\label{internal model}
\dot{\eta}_{i}=\Sigma_{i}\eta_{i}+\hat{\Sigma}_{i}e_{i},
\end{equation}
where $\eta_i$ is the internal model state driven by the local error $e_i$.

\begin{remark}\label{rem:exosystem}
Assuming the exosystem matrix $E_i$ to be known is standard in output regulation theory \citep{huang2004nonlinear}. In the present direct data-driven setting, this knowledge is required to construct the internal model in \eqref{internal model}. 
By contrast, observer-based disturbance-rejection methods typically require the plant matrices or an identified model, adaptive internal model approaches rely on online updates that are not directly compatible with offline convex synthesis \citep{BIN2019422}, and sliding-mode designs usually require a known disturbance bound \citep{loukianov2018robust}. Moreover, estimating $E_i$  from data is also challenging, since 
$w_{i}(t)$ is not directly measured and is coupled with unknown plant dynamics and measurement noise.
These requirements are not compatible with the present framework, where the controller is synthesized offline from finite noisy data with unknown system matrices. If $E_i$ is inaccurate or only partially known, unmodeled frequencies act as unmatched uncertainties, so that exact asymptotic regulation is generally replaced by approximate regulation. Extending the present framework to completely unknown exosystems would require an additional mechanism for inferring the disturbance generator from data, and is therefore left for future work.
\end{remark}

\subsection{The Stabilizing Controller}
The final component of the framework is a decentralized controller for the augmented system, which integrates the agent dynamics \eqref{linsys}, the NE model \eqref{eq:NEM}, the communication protocol \eqref{eq:cpy}, and the internal model  \eqref{internal model}. For each agent $i \in \mathcal{I}$, the control strategy is given by
\begin{align}\label{eq:ct}
u_i =K_{i}{\rm col}(x_i,\hat{y}_{-i}^i,\varsigma_i,\eta_i),
\end{align}
where $K_{i}=[K_{i}^{1} ~K_{i}^{2}~K_{i}^{3}~K_{i}^{4}]$ is the feedback gain matrix to be determined. 

These components together yield the closed-loop architecture shown in Fig. \ref{fig0}.
\begin{figure}[t]
  \centering
  \includegraphics[scale=0.45] {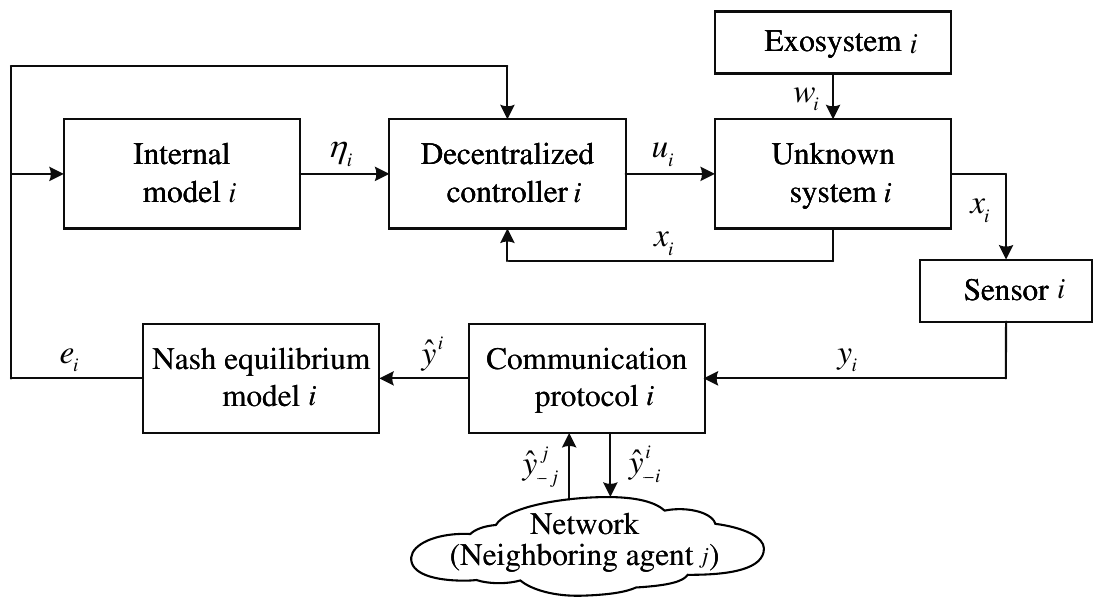}
\caption{Scheme of the closed-loop system.}\label{fig0}
  \centering
\end{figure}

\section{Decentralized Data-Driven NE Seeking for Linear Dynamics}\label{sec_lin}
The previous section established an output regulation-based framework for NE seeking and introduced the corresponding internal model-based controller. Building on this framework, the present section develops a decentralized data-driven synthesis method for the unknown linear dynamics.

To integrate the above framework into a unified augmented system, we further augment the disturbance generator with two constant components. For each $i\in\mathcal{I}$, define $\tilde{w}_{i}=\mathrm{col}(\tilde{w}_{1i}, \tilde{w}_{2i},\tilde{w}_{3i})\in \mathbb{R}^{g_{i}}$ with initial condition $\tilde{w}_{i}(0)=\mathrm{col}(w_{i}(0), 1, 1)$, which evolves according to
\begin{equation}\label{eyx_sys}
\dot{\tilde{w}}_i = \tilde{E}_i \tilde{w}_i = \mathrm{diag}(E_i, 0, 0)\tilde{w}_i.
\end{equation}
Note that the knowledge of \(E_i\) only specifies the disturbance generator; the actual disturbance trajectory \(\tilde w_i(t)\) remains unknown because it depends on the unknown initial disturbance state \(w_i(0)\).

To avoid relying on the unknown system matrices, we next perform an offline experiment to collect data for controller design. Since agent $i$ can access only local information during data collection, the inter-agent terms appearing in the communication protocol \eqref{eq:cpy} and the internal model dynamics \eqref{internal model} cannot be directly incorporated into the offline data-collection experiment. We therefore consider the following local dynamics obtained by removing these inter-agent couplings
\begin{subequations}\label{eq:argue_no}
\begin{align}
\dot{x}_i &= A_i x_i + B_i u_i +D_i w_i,\label{eq:argue_noa}\\
y_{i}&=C_{i}x_{i},\\
\dot{\breve{y}}_{-i}^i &= -\delta \sum_{j \in\mathcal{I}} a_{ij} \breve{y}_{-i}^i,\\
\dot{\varsigma}_i &= R_i y_{i} - h_i,\\
\dot{\breve{\eta}}_i&=\Sigma_{i}\breve{\eta}_i+\hat{\Sigma}_{i} \big((Q_{ii}+Q^{\top}_{ii}) y_{i} + F_{ii}^{\top} + R_i^\top \varsigma_i \big),
\end{align}
\end{subequations}
where $\breve{y}_{-i}^i$ denotes the local output estimate obtained by excluding the inter-agent terms in \eqref{eq:cpy}, and $\breve{\eta}_i$ denotes the internal model state obtained by excluding the inter-agent terms in \eqref{internal model}.
Defining the augmented local state as $\tilde{\xi}_i:= \mathrm{col}(x_i, \breve{y}_{-i}^{i}, \varsigma_i, \breve{\eta}_i)$$\in \mathbb{R}^{v_{i}}$, we obtain the augmented dynamics
\begin{equation}\label{arg_com}
\dot{\tilde{\xi}}_i = \hat{A}_i \tilde{\xi}_i + \hat{B}_i u_i + \hat{D}_i \tilde{w}_i,
\end{equation}
where
\begin{align*}
\hat{A}_i &=
\begin{bmatrix}
A_i & 0 & 0 & 0 \\
0 & -\delta \sum_{j\in\mathcal{I}} a_{ij} I & 0 & 0 \\
R_{i}C_{i} & 0 & 0 & 0 \\
\hat{\Sigma}_{i} (Q_{ii}+Q^{\top}_{ii})C_i & 0 & \hat{\Sigma}_{i} R_{i}^\top & \Sigma_{i}
\end{bmatrix},\\
\hat{B}_i &=
\begin{bmatrix}
B_i \\
0 \\
0 \\
0
\end{bmatrix}, \quad
\hat{D}_i =
\begin{bmatrix}
\tilde{D}_i \\
0 \\
 - \tilde{N}_{ii} \\
\hat{\Sigma}_{i} \tilde{F}_{ii} 
\end{bmatrix},
\end{align*}
with
$\tilde{D}_{i}=[D_i ~0_{n_{i}\times 2}]$, $\tilde{F}_{ii}=[0_{p_{i}\times q_{i}}~F_{ii}^{\top}~0_{p_{i}\times 1}]$, and $\tilde{N}_{ii}=[0_{f_{i}\times (q_{i}+1)}~h_{i}]$.

An input sequence $\{u_{i}(l)\}_{l=0}^{T-1}$ is applied to systems \eqref{eq:argue_no} to generate the trajectories $\{x_{i}(l)\}_{l=0}^{T-1}$, $\{\breve{y}_{-i}^i(l)\}_{l=0}^{T-1}$, $\{\varsigma_{i}(l)\}_{l=0}^{T-1}$, $\{\breve{\eta}_{i}(l)\}_{l=0}^{T-1}$, and $\{y_{i}(l)\}_{l=0}^{T-1}$. 
Over each interval $l\in [t_{l},t_{l+1})$, the state derivatives are approximated using forward differences, satisfying
\begin{subequations}\label{eq:coll_data}
\begin{align}
\dot{x}_i (l) &= A_i x_i (l) + B_i u_i (l) +D_i w_i (l)+\pi_{1i} (l),\label{eq:coll_data_a}\\
\dot{\breve{y}}_{-i}^i (l) &= -\delta \sum_{j \in \mathcal{I}} a_{ij}  \breve{y}_{-i}^i (l)+\pi_{2i} (l),\label{eq:coll_data_b}\\
\dot{\varsigma}_i (l) &= R_i y_i (l) - h_{i}+\pi_{3i} (l),\label{eq:coll_data_c}\\
\dot{\breve{\eta}}_i (l)&=\Sigma_{i}\breve{\eta}_i (l)+\hat{\Sigma}_{i} \Big[(Q_{ii}+Q^{\top}_{ii}) y_i (l) + F_{ii}^{\top} \nonumber\\
&\quad+ R_i^\top \varsigma_i (l)\Big]+\pi_{4i} (l)\label{eq:coll_data_d}.
\end{align}
\end{subequations}
Define $\pi_{i} (l)={\rm col}(\pi_{1i} (l),\pi_{2i} (l),\pi_{3i} (l),\pi_{4i} (l))$, where $\pi_{i}$ is the aggregated measurement noise vector arising during offline data collection.

For subsequent data-driven analysis, the collected trajectories are systematically organized into matrix form as follows
\begin{subequations}\label{Da}
\begin{align}
U_{i}^{-} &= \begin{bmatrix} u_{i}(0) & u_{i}(1) & \cdots & u_{i}(T-1) \end{bmatrix}, \\
Y_{i}^{-} &= \begin{bmatrix} y_{i}(0) & y_{i}(1) & \cdots & y_{i}(T-1) \end{bmatrix},\\
\Xi_{i}^{-} &= \begin{bmatrix} \tilde{\xi}_{i}(0) & \tilde{\xi}_{i}(1) & \cdots & \tilde{\xi}_{i}(T-1) \end{bmatrix}, \\
\Xi_{i}^{+} &= \begin{bmatrix} \dot{\tilde{\xi}}_{i}(0) & \dot{\tilde{\xi}}_{i}(1) & \cdots & \dot{\tilde{\xi}}_{i}(T-1) \end{bmatrix}.
\end{align}
\end{subequations}
Similarly, stack the unknown disturbance $\tilde{w}_{i}$ and noise $\pi_{i}$ across the $T$ time slots into the matrices
\begin{subequations}\label{Da_n}
\begin{align}
W_{i}^{-} &= \begin{bmatrix} \tilde{w}_{i}(0) & \tilde{w}_{i}(1) & \cdots & \tilde{w}_{i}(T-1) \end{bmatrix},\label{Da_n1}\\
\Pi_{i}^{-} &= \begin{bmatrix} \pi_{i}(0) & \pi_{i}(1) & \cdots & \pi_{i}(T-1) \end{bmatrix}.\label{Da_n2}
\end{align}
\end{subequations}
According to \eqref{eq:coll_data}, these data matrices satisfy the following linear relationship
\begin{equation}\label{sysX1}
\Xi_{i}^{+} = \hat{A}_i\Xi_{i}^{-} + \hat{B}_i U_{i}^{-} +\hat{D}_i W_{i}^{-}+ \Pi_{i}^{-}.
\end{equation}
To ensure that the collected data are sufficiently informative for controller synthesis, we impose the following rank condition.
\begin{assumption}\label{rank}
The matrix $\begin{bmatrix}
\begin{smallmatrix}
U_{i}^{-}\\
\Xi_{i}^{-}
\end{smallmatrix}
\end{bmatrix}$ has full row rank.
\end{assumption}

\begin{remark}\label{rankre}
Assumption~\ref{rank} is a standard informativity condition in data-driven control. In practice, this condition can be enforced by applying an input sequence that is persistently exciting of order $v_i+1$, and by choosing the data length \(T\) such that $T \ge (m_i+1)(v_i+1)-1$, where \(m_i\) is the input dimension and \(v_i\) is the dimension of the augmented local state \(\tilde{\xi}_i\). This lower bound follows from the standard length requirement for a persistently exciting input; see, e.g., \cite{de2019formulas}.
\end{remark}

To account for measurement noise $\Pi_{i}^{-}$ in the collected data, we impose the following assumption.
\begin{assumption}\label{ass_datanoise}
For each agent $i\in\mathcal{I}$, the noise matrix $\Pi_{i}^{-}$ is bounded; that is, there exists a matrix $\Delta_{i}$ such that $\Pi_{i}^{-}\in \mathbf{\Pi}_{i}$, where $\mathbf{\Pi}_{i}:=\{\Pi_{i}\in \mathbb{R}^{v_{i}\times T}:
\Pi_{i}\Pi_{i}^\top \preceq \Delta_{i} \Delta_{i}^\top\}$.

\end{assumption}
To eliminate the explicit dependence on the unavailable disturbance data matrix $W_i^-$, we next employ the factorization method in \citet{hu2024data}. Specifically, the unknown matrix $W_i^{-}$ admits the decomposition
\begin{equation}\label{deco_t}
W_{i}^{-}=\mathcal{T}_{i}\mathcal{V}_{i}\mathcal{M}_{i},
\end{equation}
where $\mathcal{T}_{i}$ is a similarity transformation matrix that brings $\tilde{E}_{i}$ into its Jordan normal form. The block-diagonal matrix $\mathcal{V}_{i}$ is defined as
\begin{align*}
\mathcal{V}_{i} &:= {\rm blkdiag}(\mathcal{V}_{a_{1i} }(\kappa_{i}(0)), \ldots, \mathcal{V}_{a_{ki} }(\kappa_{i}(0)),  \nonumber\\
&\quad\quad\mathcal{V}_{a_{(k+1)i} }(\kappa_{i}(0)),\ldots,\mathcal{V}_{a_{(k+z)i} }(\kappa_{i}(0))),
\end{align*}
and the matrix $\mathcal{M}_{i}$ is given by
\begin{equation*}
\mathcal{M}_{i} := 
\setlength{\arraycolsep}{2pt}\scriptsize{\begin{bmatrix}
\mathcal{M}_{a_{1i} }(t_{0}) &
\mathcal{M}_{a_{1i} }(t_{1}) &\cdots & \mathcal{M}_{a_{1i} }(t_{T-1}) \\
\vdots &  \ddots & \ddots & \vdots \\
\mathcal{M}_{a_{ki} }(t_{0}) & 
\mathcal{M}_{a_{ki} }(t_{1}) &\cdots & \mathcal{M}_{a_{ki} }(t_{T-1}) \\
\mathcal{M}^{*}_{a_{(k+1)i} }( t_{0}) &
\mathcal{M}^{*}_{a_{(k+1)i} }( t_{1}) &
\cdots &
\mathcal{M}^{*}_{a_{(k+1)i} }( t_{T-1}) \\
\vdots &  \ddots & \ddots &\vdots \\
\mathcal{M}^{*}_{a_{(k+z)i} } (t_{0}) &
\mathcal{M}^{*}_{a_{(k+z)i} } (t_{1}) &
 \cdots &
\mathcal{M}^{*}_{a_{(k+z)i} }(t_{T-1})
\end{bmatrix}},
\end{equation*}
where $a_{1i} + \cdots + a_{ki} + 2(a_{(k+1)i} + \cdots + a_{(k+z)i}) = g_i$, matrices $\mathcal{V}_{a_{mi} }(\kappa_{i}(0))\in \mathbb{R}^{a_{mi}\times a_{mi}}$, $\mathcal{M}_{a_{mi}}(t) : \mathbb{R}_{\ge 0} \to \mathbb{R}^{a_{mi}}$ 
associated with the real eigenvalues $\lambda_{mi}$, $m=1,\ldots,k$, and matrices 
$\mathcal{V}_{a_{(k+n)i}}(\kappa_{i}(0)) \in \mathbb{R}^{2a_{(k+n)i} \times 2a_{(k+n)i}}$, 
$\mathcal{M}^{*}_{a_{(k+n)i}}(t) : \mathbb{R}_{\ge 0} \to \mathbb{R}^{2a_{(k+n)i}}$ 
associated with complex conjugate eigenvalues $\lambda_{(k+n)i}, \lambda_{(k+n)i}^{*}$, $n=1,\ldots,z$.

As shown in \citet{hu2024data}, there exists a full-row-rank row-selection matrix $\digamma_{i}\in \mathbb{R}^{\iota_{i}\times g_{i}}$ that removes redundant block rows, yielding
\begin{equation}\label{no_deco}
\hat{\mathcal{M}}_{i}=\digamma_{i}\mathcal{M}_{i},~~\hat{\mathcal{M}}_{i}\in\mathbb{R}^{\iota_{i}\times T}.
\end{equation}
Let $\digamma_{i}^{R}\in \mathbb{R}^{g_{i}\times \iota_{i}}$ be a right inverse of $\digamma_{i}$, so that $\mathcal{M}_{i}=\digamma^{R}_{i}\hat{\mathcal{M}}_{i}$. It follows that
\begin{equation}\label{no_deco_1}
W_{i}^{-}=\mathcal{T}_{i}\mathcal{V}_{i}\digamma^{R}_{i}\hat{\mathcal{M}}_{i}.
\end{equation}
Substituting \eqref{deco_t}--\eqref{no_deco_1} into \eqref{sysX1} yields
\begin{equation}\label{sysX1_de}
\Xi_{i}^{+} = \hat{A}_i\Xi_{i}^{-} + \hat{B}_i U_{i}^{-} +\hat{D}_i \mathcal{T}_{i}\mathcal{V}_{i}\digamma^{R}_{i}\hat{\mathcal{M}}_{i}+ \Pi_{i},~ \forall\Pi_{i}\in\mathbf{\Pi}_{i}.
\end{equation}
The following lemma, adapted from \cite[Theorem 2]{de2019formulas}, provides a data-driven representation of the closed-loop system dynamics.
\begin{lem}[\citealp{de2019formulas}]\label{lem_ma}
Let matrix $G_i \in \mathbb{R}^{T \times v_i}$ satisfy
\begin{equation}\label{eq:lemma1}
\begin{bmatrix}
K_i \\
I_{v_i} \\
0_{\iota_{i} \times v_i}
\end{bmatrix}
=
\begin{bmatrix}
U_{i}^{-} \\
\Xi_{i}^{-} \\
\hat{\mathcal{M}}_{i}
\end{bmatrix}
G_i. 
\end{equation}
Then, the closed-loop dynamics of system \eqref{arg_com} under the decentralized control strategy \eqref{eq:ct} is given by
\begin{equation}
\begin{aligned}
\dot{\tilde{\xi}}_i 
&= (\hat{A}_i+\hat{B}_i K_{i})\tilde{\xi}_i + \hat{D}_{i} \tilde{w}_{i}\\
&= (\Xi_{i}^{+}-\Pi_{i}) G_i \tilde{\xi}_i + \hat{D}_{i} \tilde{w}_{i}.
 \end{aligned}
\end{equation}
\end{lem}
The following theorem shows that the proposed decentralized control strategy \eqref{eq:ct} solves the output regulation problem and hence the NE seeking Problem \ref{pro_1}, without requiring knowledge of system models. 
\begin{theorem}\label{lin_the}
Consider the system \eqref{linsys} under Assumptions \ref{no_eig}--\ref{ass_pseudo}.
For data matrices $U_{i}^{-}$, $Y_{i}^{-}$, $\Xi_{i}^{-}$, $\Xi_{i}^{+}$,  $W_{i}^{-}$, $\Pi_{i}^{-}$ in \eqref{Da}-\eqref{Da_n} satisfying Assumptions \ref{rank}--\ref{ass_datanoise}, 
if the SDPs in \eqref{condi1} are feasible for some matrices $P_i \succ 0$, $S_i$, and scalars $\tau_i>0$, $\chi_i>0$, $\alpha_i>0$, then the decentralized  control strategy  \eqref{eq:ct} with  $K_{i}= U_{i}^{-}S_i(\Xi_{i}^{-}S_{i})^{-1}$ solves  Problem \ref{pro_1}. 
\begin{subequations}\label{condi1}
\begin{align}
&  P_i \succeq \tau_i I_{v_i},\label{condi1_2}\\
&\tau_i > 2\alpha_i\gamma_i,\label{condi1_3}\\
&\setlength{\arraycolsep}{2pt}\begin{bmatrix}
\Xi_{i}^{+}S_{i}+(\Xi_{i}^{+}S_{i})^{\top}+\chi_i\Delta_{i} \Delta_{i}^{\top}  & S_{i}^{\top} & P_{i}^{\top}\\
S_{i} &-\chi_i I_{T} & 0\\
P_{i} & 0 &- \alpha_{i}I_{v_{i}} \end{bmatrix} \preceq 0,\label{condi1_1}\\
&P_{i}=\Xi_{i}^{-}S_{i} \label{condi1_5}\\
&0_{\iota_{i} \times v_{i}}  =\hat{\mathcal{M}}_{i} S_{i},\label{condi1_4}
\end{align}
\end{subequations}
where $\gamma_i=\sum_{j \in \mathcal{I}} \|A_{\mathrm{cl,ij}}\|$, $i\in \mathcal{I}$. 
\end{theorem}

\begin{pf}
For each agent $i \in \mathcal{I}$, define the augmented state $\xi_i := \operatorname{col}(x_i, \hat{y}_{-i}^i, \varsigma_i, \eta_i)$,
and stack $\xi := \operatorname{col}(\xi_1, \ldots, \xi_N)$, $\tilde{w} := \operatorname{col}(\tilde{w}_1, \ldots, \tilde{w}_N)$, 
$\tilde{E} := \operatorname{blkdiag}(\tilde{E}_1, \ldots, \tilde{E}_N)$.
Then, by combining the agent dynamics \eqref{linsys}, the NE model \eqref{eq:NEM}, the communication protocol \eqref{eq:cpy}, the internal model \eqref{internal model}, and the controller \eqref{eq:ct},
the closed-loop system can be written in the standard output regulation form
\begin{subequations}\label{re_eq}
\begin{align}
\dot{\xi} &= A_{\mathrm{cl}}\xi + D_{\mathrm{cl}}\tilde{w}, \label{re_eqa}\\
\dot{\tilde{w}} &= \tilde{E}\tilde{w},\label{re_eqc}\\
e &= C_{\mathrm{cl}} \xi + F_{\mathrm{cl}} \tilde{w},\label{re_eqb}
\end{align}
\end{subequations}
for suitable matrices $A_{\mathrm{cl}}, D_{\mathrm{cl}}, C_{\mathrm{cl}}, F_{\mathrm{cl}}$,
where $A_{\mathrm{cl}}$ is a block matrix with diagonal blocks
$A_{\mathrm{cl,i}}$ and off diagonal blocks $A_{\mathrm{cl,ij}}$ for $i,j\in\mathcal{I}$, 
\begin{align*}
{A}_{\mathrm{cl,i}} &=
\begin{bmatrix}
A_i+B_{i}K_{i}^{1} & B_{i}K_{i}^{2} & B_{i}K_{i}^{3} & B_{i}K_{i}^{4} \\
0 & -\delta \sum_{j\in\mathcal{I}} a_{ij} I & 0 & 0 \\
R_{i}C_{i} & 0 & 0 & 0 \\
\hat{\Sigma}_{i} (Q_{ii}+Q^{\top}_{ii})C_i & 0 & \hat{\Sigma}_{i} R_{i}^\top & \Sigma_{i}
\end{bmatrix},\\
A_{\mathrm{cl,ij}} &=
\begin{bmatrix}
0 & 0 & 0 & 0 \\
\delta a_{ij}\mathbf{H}_{y_i}\mathbf{M}_{y_j}^\top C_j & \delta a_{ij} \mathbf{H}_{y_i} \mathbf{H}_{y_j}^\top & 0 & 0 \\
0 & 0 & 0 & 0 \\
\hat{\Sigma}_{i} Q_{ij} C_j & 0 & 0 & 0
\end{bmatrix}.
\end{align*}


By the classical linear output regulation theorem \cite[Theorem 1.31]{huang2004nonlinear}, under Assumption \ref{no_eig}, it is sufficient to prove that (i) $A_{\mathrm{cl}}$ is Hurwitz; and (ii) the regulator equations 
\begin{subequations}\label{re_XA}
\begin{align}
\tilde{X} \tilde E = A_{\mathrm{cl}}\tilde{X} + D_{\mathrm{cl}}, \label{re_xa}\\ 
\quad 0 = C_{\mathrm{cl}}\tilde{X} + F_{\mathrm{cl}}\label{re_xb},
\end{align}
\end{subequations}
have a unique solution $\tilde{X}$.

We first prove that $A_{\mathrm{cl}}$ is Hurwitz. If \eqref{condi1} is feasible, we can define $G_{i}=S_i P_i^{-1}$. Then, \eqref{condi1_5} and \eqref{condi1_4}, together with the feedback gain $K_{i}= U_{i}^{-}S_i(\Xi_{i}^{-}S_{i})^{-1}$ ensure that condition \eqref{eq:lemma1} is satisfied, thereby guaranteeing consistency with Lemma \ref{lem_ma}.

Applying the Schur complement lemma to \eqref{condi1_1} yields
\begin{equation}\label{th1_1}
\Xi_{i}^{+}S_{i}+(\Xi_{i}^{+}S_{i})^{\top}+\frac{1}{\alpha_{i}}P_{i}^{\top}P_{i}+\chi_i\Delta_{i} \Delta_{i}^{\top} +\frac{1}{\chi_i}S_{i}^{\top}S_{i} \preceq 0.
\end{equation}
This inequality can be expressed more compactly as
\begin{equation}\label{lin_pet}
\Psi_{i}+\chi_i\Phi_{i}^{\top}\Delta_{i} \Delta_{i}^{\top}\Phi_{i}+\frac{1}{\chi_i}\Theta_{i}\Theta_{i}^{\top}\preceq 0,
\end{equation}
where $\Psi_{i}:=\Xi_{i}^{+}S_{i}+(\Xi_{i}^{+}S_{i})^{\top}+\frac{1}{\alpha_{i}}P_{i}^{\top}P_{i}$, $\Phi_{i}:=-I_{v_{i}}$, and $\Theta_{i}:=S_{i}^{\top}$. 
Under Assumption \ref{ass_datanoise}, the nonsingular Petersen lemma yields
\begin{equation}
\Psi_{i}+\Theta_{i}\Pi_{i}^{\top}\Phi_{i}+\Phi_{i}^{\top}\Pi_{i}\Theta_{i}^{\top}\preceq 0,
\end{equation}
for all $\Pi_{i}\in \mathbf{\Pi}_{i}$, which is equivalent to
\begin{equation}\label{th1_3}
(\Xi_{i}^{+}-\Pi_{i})S_{i}+ S_{i}^\top(\Xi_{i}^{+}-\Pi_{i})^{\top} \preceq -\frac{1}{\alpha_{i}}P_{i}^{\top}P_{i}.
\end{equation}
By Lemma \ref{lem_ma}, we have
\begin{equation}\label{th1_2}
{A}_{\mathrm{cl,i}}P_{i} +P_{i}^{\top}{A}_{\mathrm{cl,i}}^{\top}  \preceq -\frac{1}{\alpha_{i}}P_{i}^{\top}P_{i}.
\end{equation}
Consider the Lyapunov candidate 
$V_i(\xi_i) = \xi_i^\top \Gamma_{i} \xi_i$, where $\Gamma_{i}=P_{i}^{-1}$.
Using inequality \eqref{condi1_2}, the derivative of $V_{i}$ along the nominal local dynamics
\begin{align}
&\nabla V_i(\xi_i)^{\top}(\hat{A}_{i}+\hat{B}_{i}K_{i})\xi_i\\
&\quad=\xi_i^\top \!\left(\Gamma_{i}A_{{\rm cl},i}+A_{{\rm cl},i}^\top\Gamma_{i}\right)\!\xi_i \le-\frac{1}{\alpha_i}\|\xi_i\|^2.\nonumber
\end{align}
Furthermore, the lower bound condition in \eqref{condi1_2} implies
\begin{equation}
    \|\nabla V_i(\xi_i)\|=2\|\Gamma_{i}\xi_i\| \leq \frac{2}{\tau_i} \|\xi_i\|.
\end{equation}
We now account for the restored online couplings. For the interaction terms,
\begin{equation}\label{sumj}
\|\sum_{j \in \mathcal{I}} A_{\mathrm{cl,ij}}\xi_j\| \leq \sum_{j \in \mathcal{I}} \|A_{\mathrm{cl,ij}}\| \|\xi_j\|.
\end{equation}
Define the matrix $O = [o_{ij}]$ with entries
\begin{align*}
o_{ij} =
\begin{cases}
\frac{1}{\alpha_i}, & i = j, \\
-\frac{2}{\tau_i} \|A_{\mathrm{cl,ij}}\|, & i \neq j.
\end{cases}
\end{align*}
Under  \eqref{condi1_3}, the matrix $O$ is strictly diagonally dominant with nonpositive off-diagonal entries, which is therefore a nonsingular M-matrix.
Thus, the origin of the homogeneous stacked system  $\dot \xi = A_{\mathrm{cl}}\xi$ is exponentially stable \cite[Theorem 9.2]{khalil2002nonlinear}. Since this system is linear time-invariant, exponential stability of the origin is equivalent to $A_{\mathrm{cl}}$ being Hurwitz.

We now turn to the regulator equations. Since the stacked controller state embeds the internal model associated with the exosystem $\tilde E$, and since $A_{cl}$ has been shown to be Hurwitz, according to \cite[Lemma 1.27]{huang2004nonlinear} and \cite[Theorem 1]{guo2021linear}, the standard internal model argument ensures that the steady-state response of the closed-loop system is well defined and satisfies the corresponding regulator equations  \eqref{re_XA}. Define the shifted variable $\bar\xi:=\xi-\tilde X\tilde w$. Using \eqref{re_eqa}, \eqref{re_eqc}, and \eqref{re_xa}, one obtains $\dot{\bar\xi}=A_{cl}\bar\xi$. Moreover, by \eqref{re_eqb} and \eqref{re_xb}, $e=C_{cl}\bar\xi$. Since \(A_{cl}\) is Hurwitz, it follows that $\bar\xi(t)\to0$, $e_i(t)\to0$, $\forall i\in\mathcal I$. Furthermore, the communication protocol yields asymptotically consistent estimates, i.e., $\hat y^i_{-i}(t)-y_{-i}(t)\to0$. Together with the multiplier dynamics $\dot\varsigma_i=R_i y_i-h_i$, this implies that the limiting closed-loop steady state satisfies the KKT residual conditions associated with the constrained game. Hence, under Assumption \ref{ass_pseudo}, the limiting output coincides with the unique constrained NE \(y^\ast\). Therefore, $y(t)\to y^\ast$, and the decentralized controller solves Problem~\ref{pro_1}.

\end{pf}

\begin{remark}
Although the controller is synthesized from the decoupled system \eqref{arg_com}, the actual closed-loop dynamics \eqref{re_eqa} are coupled. The quantity $\gamma_i=\sum_{j\in\mathcal I}\|A_{cl,ij}\|$ provides a conservative norm-based upper bound on the restored interconnection effect acting on agent~$i$. Hence, the condition $\tau_i > 2\alpha_i\gamma_i$ in \eqref{condi1_3} ensures that the local stabilization margin dominates the interconnection effect.
Moreover, $\gamma_i$ can be estimated in a data-driven way. Since $Y_j=C_jX_j$ and $X_j$ is full row rank under the informativity condition, one has $C_j=Y_jX_j^\dagger$, so that $\|C_j\|$ and thus $\gamma_i$ can be computed using only local data exchange. The resulting bound is conservative due to its norm-based construction, which is inherent to the fully decentralized synthesis.
Finally, stronger couplings generally increase $\gamma_i$, making $\tau_i > 2\alpha_i\gamma_i$ harder to satisfy and shrinking the feasible SDP region. Since the guaranteed decay rate also depends on $\alpha_i$, stronger couplings may require a more conservative choice of $\alpha_i$, leading to a slower guaranteed convergence rate.
\end{remark}

Theorem \ref{lin_the} establishes that an internal model-based controller guarantees NE seeking under disturbances generated by \eqref{disturbance}. We now consider two fundamental special cases: (i) constant disturbances and (ii) no disturbances, both corresponding to an exosystem \eqref{disturbance}  with $E_i = 0_{q_i \times q_i}$. In this setting, the internal model in \eqref{internal model} reduces to an integral control structure
\begin{subequations}\label{control_strategy_coll}
\begin{align}
u_i &=K_{i}{\rm col}(x_i,\hat{y}_{-i}^i,\varsigma_i,\varphi_{i}),\label{state_control_coll}\\
\dot{\varphi}_{i} &= e_{i}, \label{internal model_coll}
\end{align}
\end{subequations}
where $\varphi_{i}\in\mathbb{R}^{p_{i}}$ is the integral state. 
The following corollary extends Theorem \ref{lin_the} by resolving Problem \ref{pro_1} for these cases.

\begin{corollary}\label{corollary}
Consider the system \eqref{linsys} under Assumptions \ref{no_eig}--\ref{ass_pseudo}.
For data matrices $U_{i}^{-}$, $Y_{i}^{-}$, $\Xi_{i}^{-}$, $\Xi_{i}^{+}$,  $W_{i}^{-}$, $\Pi_{i}^{-}$ in \eqref{Da}-\eqref{Da_n} satisfying Assumptions \ref{rank}--\ref{ass_datanoise}, 
if the SDPs in \eqref{condi1} are feasible for some matrices $P_i \succ 0$, $S_i$, and scalars $\tau_i>0$, $\chi_i>0$, $\alpha_i>0$, then the decentralized  control strategy  \eqref{state_control_coll} with  $K_{i}= U_{i}^{-}S_i(\Xi_{i}^{-}S_{i})^{-1}$ solves  Problem \ref{pro_1}. 
\end{corollary}

\begin{remark}
For agents subject to constant disturbances or no disturbances, the $ p_{i}$-copy internal model $(\Sigma_{i},\hat{\Sigma}_{i})$ in strategy \eqref{eq:ct} can be replaced by $(0_{p_{i}\times p_{i}},I_{p_{i}})$ as in controller \eqref{control_strategy_coll}. 
In this case, integral control alone is sufficient to reject constant disturbances, eliminating the need for a higher-order internal model and thus reducing implementation complexity.
The constant-disturbance scenario considered in  \citet{romano2025game} is therefore a special case of the proposed general framework.
\end{remark}

\section{Decentralized Data-Driven NE Seeking for Nonlinear Dynamics}\label{sec_non}
The preceding section developed an internal model-based data-driven design for the linear case of equality-constrained NE seeking under partial-decision information. We now extend this framework to a class of nonlinear systems. Since tractable exact regulation is generally difficult to achieve for general nonlinear systems in the present data-driven setting, especially in the presence of equality constraints and partial-decision information, we focus on systems with constant disturbances and adopt the integral-control structure in Corollary \ref{corollary}. This yields a nonlinear extension of the present data-driven framework, in contrast to the model-based setting considered in \cite{romano2025game}.

Consider $N$ agents with nonlinear dynamics
\begin{equation}\label{nonlinsys}
\begin{aligned}
\dot{x}_i &= \mathcal{A}_{i}\mathcal{Z}(x_{i}) + B_iu_i + D_i w_i,\\
y_i &=  \mathcal{C}_{i}\mathcal{Z}(x_{i}), \quad i\in\mathcal{I},
\end{aligned}
\end{equation}
where $x_{i}$, $u_{i}$, and $y_{i}$ are as in Section \ref{sec_pre} and $w_{i}$ is a constant disturbance. The function $\mathcal{Z}(x_{i})={\rm col}(x_{i},Q(x_{i}))$ stacks the state and nonlinear components $Q(x_{i}) : \mathbb{R}^{n_{i}} \rightarrow \mathbb{R}^{h_{i}-n_{i}} $.  
The matrices  $\mathcal{A}_{i} \in \mathbb{R}^{{n_{i}} \times h_{i}}$, $B_i \in \mathbb{R}^{n_i \times m_i}$, $\mathcal{C}_{i} \in \mathbb{R}^{p_{i} \times h_{i}}$,  $D_i \in \mathbb{R}^{n_i \times q_i}$ are unknown.
Consistent with the partition of $Z(x_i)$, we partition
    $\begin{bmatrix}
\begin{smallmatrix}
\mathcal{A}_{i} \\ \mathcal{C}_{i}
\end{smallmatrix}\end{bmatrix}=\begin{bmatrix}
\begin{smallmatrix}
\bar{\mathcal{A}}_{i} &\tilde{\mathcal{A}}_{i} \\ \bar{\mathcal{C}}_{i} &\tilde{\mathcal{C}}_{i}
\end{smallmatrix}\end{bmatrix}$.

With the integral controller of Corollary \ref{corollary}, the design reduces to synthesizing decentralized stabilizing gains for the augmented system via data-driven LMIs. We adopt the decentralized control strategy
\begin{equation}\label{nonstate_control}
u_{i} =\mathcal{K}_{i} {\rm col}(x_i,\hat{y}_{-i}^i,\varsigma_i,\varphi_i,Q(x_{i})),
\end{equation}
where $\mathcal{K}_{i}=[\mathcal{K}_{i}^{1}~~ \mathcal{K}_{i}^{2}~~ \mathcal{K}_{i}^{3}~~ \mathcal{K}_{i}^{4}~~\mathcal{K}_{i}^{5}]$ denotes the decentralized control gain to be designed.

We now formalize the NE seeking problem for agents with unknown nonlinear dynamics.
\begin{problem}\label{pro_2}
Under Assumptions \ref{ass_graph} and \ref{ass_pseudo}, consider the equality-constrained game \eqref{min_cost} for $N$ agents governed by the unknown nonlinear dynamics \eqref{nonlinsys} and subject to constant disturbances $w_i$.
In the partial-decision information setting, the goal is to design a decentralized control strategy $u_{i}$ for each agent $i\in\mathcal{I}$ such that
(i) all closed-loop signals are bounded for all $t\geq0$; and
(ii) the output $y$ converges to the NE $y^*$.
\end{problem}

Analogous to the linear case, we excite the following local augmented dynamics
\begin{subequations}\label{eq:argue_non}
\begin{align}
\dot{x}_i &= \mathcal{A}_i \mathcal{Z}(x_i)  + B_i u_i  +\tilde{D}_{i} \tilde{w}_i,\label{eq:argue_a}\\
\dot{\breve{y}}_{-i}^i  &= -\delta \sum_{j \in \mathcal{I}} a_{ij}  \breve{y}_{-i}^i, \\
\dot{\varsigma}_i  &= R_iy_{i} - h_i, \\
\dot{\breve{\varphi}}_i &=(Q_{ii}+Q^{\top}_{ii}) y_{i} + F_{ii}^{\top} + R_i^\top \varsigma_i,
\end{align}
\end{subequations}
where $\breve{\varphi}_{i}$ denotes the integral state obtained by
excluding inter-agent terms of  \eqref{internal model_coll}, and $\tilde{w}_i = \text{col}(w_i, 1, 1)$.
Applying a persistently exciting input sequence $ \{u_{i} (l)\}_{l=0}^{T-1} $ to \eqref{eq:argue_non} yields trajectories $ \{x_{i} (l)\}_{l=0}^{T-1} $, $\{\breve{y}_{-i}^{i}(l)\}_{l=0}^{T-1}$, $ \{\varsigma_{i} (l)\}_{l=0}^{T-1} $, $ \{\breve{\varphi}_{i} (l)\}_{l=0}^{T-1} $, $ \{y_{i} (l)\}_{l=0}^{T-1} $, and $\{\pi_{i}(l)\}_{l=0}^{T-1}$, where $\pi_{i}$ is the measurement noise in data collection as defined in Section \ref{sec_lin}.
Further, the forward difference is utilized to obtain the corresponding derivative terms.

Define the augmented state vectors as
\begin{align}
 \breve{z}_i&={\rm col}(x_{i},\breve{y}_{-i}^i, \varsigma_i,\breve{\varphi}_i)\in \mathbb{R}^{\zeta_i}, \\
 \breve{\mathcal{Z}}_{i}&={\rm col}(x_{i},\breve{y}_{-i}^i, \varsigma_i,\breve{\varphi}_i,Q(x_{i}))\in \mathbb{R}^{\varpi_i}.
\end{align}
The collected data trajectories are organized into matrices as follows
\begin{subequations}\label{datamatrix}
\begin{align}
Z_{i}^{-} &=\begin{bmatrix} \breve{\mathcal{Z}}_{i}(0) & \breve{\mathcal{Z}}_{i}(1) & \ldots & \breve{\mathcal{Z}}_{i}(T-1) \end{bmatrix},\label{datamatrix_z-}\\
Z_{i}^{+} &= \begin{bmatrix} \dot{\breve{z}}_{i}(0) & \dot{\breve{z}}_{i}(1) & \ldots & \dot{\breve{z}}_{i}(T-1)\end{bmatrix},\label{datamatrix_z+}
\end{align}
\end{subequations}
which satisfy
\begin{equation}\label{nondatasystem}
Z_{i}^{+} = \hat{\mathcal{A}}_{i}Z_{i}^{-} + \hat{B}_iU_{i}^{-} + \hat{\mathcal{D}}_{i}W_{i}^{-}+\Pi_{i}^{-},
 \end{equation}
where
\begin{align*}
\hat{\mathcal{A}}_{i} &=\setlength{\arraycolsep}{2pt}\scriptsize{
\begin{bmatrix}
\bar{\mathcal{A}}_i & 0 & 0 & 0 & \tilde{\mathcal{A}}_{i}\\
0 & -\delta \sum_{j\in\mathcal{I}} a_{ij} I & 0 & 0 & 0 \\
R_{i}\bar{\mathcal{C}}_{i} & 0 & 0 & 0 & R_{i}\tilde{\mathcal{C}}_i \\
 (Q_{ii}+Q^{\top}_{ii})\bar{\mathcal{C}}_i & 0 &  R_{i}^\top & 0 &(Q_{ii}+Q^{\top}_{ii})\tilde{\mathcal{C}}_i
\end{bmatrix}},
\end{align*}
$\hat{\mathcal{D}}_{i} =
\begin{bmatrix}
\tilde{D}_i^{\top} &
0 &
 - \tilde{N}_{ii}^{\top} &
\tilde{F}_{ii}^{\top} 
\end{bmatrix}^{\top}$, $\hat{B}_{i}$ is defined in \eqref{arg_com},
and data matrices $U_{i}^{-}$, $W_{i}^{-}$, and $\Pi_{i}^{-}$ are defined in \eqref{Da}-\eqref{Da_n}.

We restate Assumption \ref{rank} for the nonlinear setting.
\begin{assumption}\label{rank_non}
The matrix 
$\begin{bmatrix}
\begin{smallmatrix}
U_{i}^{-} \\ Z_{i}^{-}
\end{smallmatrix}
\end{bmatrix}$
has full row rank.
\end{assumption}

In light of the analysis in Section \ref{sec_lin} and the results in \citet{hu2024data},
the unknown constant disturbance data matrix $W_{i}^{-}$  can be expressed as
\begin{align}
W_{i}^{-}=\check{\mathcal{V}}_{i}\check{\digamma}_{i}^{R}\check{\mathcal{M}}_{i},
\end{align}
where $\check{\mathcal{V}}_{i}:=\text{blkdiag}\{\tilde{w}_{i1}, \ldots,\tilde{w}_{ig_{i}}\}\in \mathbb{R}^{g_{i}\times g_{i}}$, $\check{\digamma}_{i}^{R}\in \mathbb{R}^{ g_{i}\times 1}$, and $\check{\mathcal{M}}_{i}=\mathbf{1}_{1\times T}$.
Substituting this into \eqref{nondatasystem} yields
\begin{equation}\label{nondatasystem_frac}
Z_{i}^{+} = \hat{\mathcal{A}}_{i}Z_{i}^{-} + \hat{B}_{i}U_{i}^{-} + \hat{\mathcal{D}}_{i}\check{\mathcal{V}}_{i}\check{\digamma}_{i}^{R}\check{\mathcal{M}}_{i}+\Pi_{i},~ \forall\Pi_{i}\in\mathbf{\Pi}_{i}.
\end{equation}
According to \cite[Lemma 3]{hu2024data}, Lemma \ref{lem_ma} can be extended to the nonlinear case as follows.
\begin{lem}\label{lem_ma_no}
Let the matrix $\mathcal{G}_i \in \mathbb{R}^{T \times \varpi_i}$ satisfy
\begin{equation}\label{eq:lemma2}
\begin{bmatrix}
K_i \\
I_{\varpi_i} \\
0_{1 \times \varpi_i}
\end{bmatrix}
=
\begin{bmatrix}
U_{i}^{-} \\
Z_{i}^{-} \\
\check{\mathcal{M}}_i
\end{bmatrix}
\mathcal{G}_i. 
\end{equation}
Let $\mathcal{G}_i$ be partitioned as $\mathcal{G}_i=[\mathcal{G}_{1i}~~\mathcal{G}_{2i}]$, where $\mathcal{G}_{1i}\in \mathbb{R}^{T \times \zeta_{i}}$ and $\mathcal{G}_{2i}\in \mathbb{R}^{T \times (\varpi_i-\zeta_{i})}$.
Then, under the decentralized control strategy \eqref{nonstate_control}, 
system \eqref{eq:argue_non} admits the closed-loop dynamics
\begin{align}
\dot{\breve{z}}_{i}&=(\hat{\mathcal{A}}_{i} + \hat{B}_{i}\mathcal{K}_{i})\breve{\mathcal{Z}}_{i} + \hat{\mathcal{D}}_{i} \tilde{w}_{i}\\
&= (Z_{i}^{+}-\Pi_{i}) \mathcal{G}_{1i} \breve{z}_{i}+(Z_{i}^{+}-\Pi_{i}) \mathcal{G}_{2i} Q(x_{i})  + \hat{\mathcal{D}}_{i}\tilde{w}_{i}.\nonumber
\end{align}
\end{lem}
To restrict the admissible nonlinearities in a form compatible with the S-procedure and LMI analysis, we impose the following quadratic constraint on the augmented state and the nonlinear term.

\begin{assumption}\label{nonlin_ass}
 For each $i\in\mathcal I$, the state $\check z_i$ and the nonlinearity $Q(x_i)$ satisfy
\begin{equation}\label{non_cons}
\begin{bmatrix}
\breve{z}_{i} \\ Q(x_{i})
\end{bmatrix}^{\top}
\begin{bmatrix}
\mathcal{Q}_{i} & \mathcal{X}_{i}  \\
\mathcal{X}_{i} ^{\top} & \mathcal{R}_{i} 
\end{bmatrix}
\begin{bmatrix}
\breve{z}_{i} \\ Q(x_{i})
\end{bmatrix} \succeq 0,
\end{equation}
 where $\mathcal{Q}_{i} = \mathcal{Q}_{i}^{\top} \succeq 0\in \mathbb{R}^{\zeta_{i}\times \zeta_{i}}$, $\mathcal{R}_{i}  \prec 0\in \mathbb{R}^{(\varpi_i-\zeta_{i})\times(\varpi_i- \zeta_{i})}$, and $\mathcal{X}_{i}\in \mathbb{R}^{\zeta_{i}\times(\varpi_i-\zeta_{i}) } $ are known matrices.
\end{assumption}

\begin{remark}
Compared with the linear case, the nonlinear extension introduces two main challenges. The first is nonlinear data representation, since the direct data parametrization available for linear dynamics is no longer applicable once nonlinear terms are present. In this paper, this issue is addressed through the lifted representation $Z(x_i)=\operatorname{col}(x_i,Q(x_i))$. The second is controller design and stability analysis. In the nonlinear setting, tractable synthesis and verifiable stability conditions are substantially harder to obtain, especially under noisy data and distributed couplings. To ensure tractable SDP conditions, Assumption \ref{nonlin_ass} imposes quadratic constraints on the nonlinearities. 
While structured, this quadratic constraint setting still covers a meaningful class of nonlinearities, including Lipschitz/norm-bounded, sector-bounded, and monotone nonlinearities, which effectively capture real-world physical phenomena such as trigonometric dynamics in mechanical and aerospace systems, actuator saturation, and friction.
\end{remark}

With this constraint in place, the following theorem demonstrates that the decentralized control strategy \eqref{nonstate_control} solves Problem \ref{pro_2} in the nonlinear setting.

\begin{theorem}\label{nonlin_the}
Consider the system \eqref{nonlinsys} under Assumptions \ref{ass_graph}--\ref{ass_pseudo}. 
Let the data matrices $U_{i}^{-}$, $Z_{i}^{-}$, $Z_{i}^{+}$,  $W_{i}^{-}$, $\Pi_{i}^{-}$ in \eqref{datamatrix} satisfy Assumptions \ref{ass_datanoise}--\ref{rank_non}, and assume the nonlinearity satisfies Assumption \ref{nonlin_ass}. 
If the SDPs in \eqref{NONSDP} are feasible for some matrices 
$\mathcal{P}_{i}\succ 0$, $ \mathcal{S}_{i}$, $ \mathcal{G}_{2i}$,
and scalars  $\tau_i>0$,  $\chi_i>0$, $ \alpha_i >0 $, then the decentralized  control strategy  \eqref{nonstate_control} with  $\mathcal{K}_{i} = U_{i}^{-} 
\begin{bmatrix} \mathcal{S}_{i}\mathcal{P}_{i}^{-1} & \mathcal{G}_{2i}\end{bmatrix}$ solves  Problem \ref{pro_2}. 
\begin{subequations}\label{NONSDP}
\begin{align}
&\mathcal{P}_{i} \succeq \tau_i I_{\zeta_i},\label{NONcondi1_2}\\
&\tau_i > 2\alpha_i\gamma_i,\label{NONcondi1_3}\\
&\setlength{\arraycolsep}{2pt}\tiny{\begin{bmatrix}
  M(\mathcal{S}_{i},\chi_i) & Z_{i}^{+}\mathcal{G}_{2i}+\mathcal{P}_{i}\mathcal{X}_{i} & \mathcal{P}_{i}\mathcal{Q}_{i}^{1/2} &\mathcal{P}_{i}& \mathcal{S}_{i}^{\top}\\
  \star & \mathcal{R}_{i} & 0_{(\varpi_i-\zeta_i) \times \zeta_i}&0_{(\varpi_i-\zeta_i) \times \zeta_i}& \mathcal{G}_{2i}^{\top}\\
  \star & \star & -I_{\zeta_i}&0_{\zeta_i \times \zeta_i}&0_{\zeta_i \times T}\\
  \star&\star&\star &-\alpha_i I_{\zeta_{i}} &0_{ \zeta_i\times T}\\
  \star &\star&\star&\star&-\chi_i I_{T}\\
  \end{bmatrix}\!\preceq 0},\label{NONSDPc}\\
  &Z_{i}^{-}\mathcal{S}_{i} = \begin{bmatrix} \mathcal{P}_{i} \\ 0_{(\varpi_i-\zeta_i) \times \zeta_i} \end{bmatrix},\label{NONSDPb}\\
  &Z_{i}^{-}\mathcal{G}_{2i} = \begin{bmatrix} 0_{\zeta_i \times (\varpi_i-\zeta_i)} \\ I_{\varpi_i-\zeta_i} \end{bmatrix},\label{NONSDPd}\\
  &\mathbf{1}_{1 \times T}\begin{bmatrix}\mathcal{S}_{i} &\mathcal{G}_{2i}\end{bmatrix}=0_{1\times \varpi_i},\label{NONSDPe}
\end{align}
\end{subequations}
where $M(\mathcal{S}_{i},\chi_i)=Z_{i}^{+}\mathcal{S}_{i} +(Z_{i}^{+}\mathcal{S}_{i})^\top+\chi_i\Delta_{i}\Delta_{i}^{\top}$,
$\gamma_i=\sum_{j \in \mathcal{I}} \|\mathcal{A}_{\mathrm{cl,ij}}\|$,
$i\in \mathcal{I}$.
\end{theorem}
\begin{pf}
The proof proceeds along the same lines as that of Theorem \ref{lin_the}. We therefore only emphasize the additional arguments required by the nonlinear lifted term $Q(x_{i})$.  
Let $\mathcal{G}_{1i} = \mathcal{S}_{i} \mathcal{P}_{i}^{-1}$. 
By \eqref{NONSDPb}--\eqref{NONSDPe} and the definition of $\mathcal{K}_{i} = U_{i}^{-} 
\begin{bmatrix} \mathcal{S}_{i}\mathcal{P}_{i}^{-1} & \mathcal{G}_{2i}\end{bmatrix}$, the condition \eqref{eq:lemma2} in Lemma \ref{lem_ma_no} is satisfied.

Next, applying the Schur complement and the nonsingular Petersen lemma to \eqref{NONSDPc} yields, for all admissible
$\Pi_{i}\in\mathbf{\Pi}_{i}$,
\begin{align}\label{nos_cons}
\begin{bmatrix}
\hat{M}(\mathcal{P}_{i})+\frac{1}{\alpha_{i}}\mathcal{P}_{i}^{2}+\mathcal{P}_{i}\mathcal{Q}_{i}\mathcal{P}_{i} & (Z_{i}^{+}-\Pi_{i})\mathcal{G}_{2i}+P_{i}\mathcal{X}_{i}\\
  \star& \mathcal{R}_{i} 
  \end{bmatrix}\preceq 0,
\end{align}
where $\hat{M}(\mathcal{P}_{i})=(Z_{i}^{+}-\Pi_{i})\mathcal{P}_{i}+\mathcal{P}_{i}(Z_{i}^{+}-\Pi_{i})^{\top}$.
Pre- and post-multiplying the above inequality by ${\rm col}( \mathcal{P}_{i}^{-1}z_i,  Q(x_{i}))$, and then invoking Assumption \ref{nonlin_ass} through the S-procedure, we obtain
\begin{align}\label{nonlmi_6}
z_{i}^{\top}\mathcal{P}_{i}^{-1}\mathcal{A}_{{\rm cl},i}\mathcal{Z}_i+
\mathcal{Z}_{i}^{\top}\mathcal{A}_{{\rm cl},i} ^\top\mathcal{P}_{i}^{-1} z_i \le -\frac{1}{\alpha_i} z_i^{\top}z_i.
\end{align}
Consider the Lyapunov function
$V_i(z_i):=z_i^\top\mathcal P_i^{-1}z_i$.
Define \(\mathcal O=[\tilde o_{ij}]\) analogously to \(O=[o_{ij}]\) in the proof of Theorem~\ref{lin_the}, with
\(A_{{\rm cl},ij}\) replaced by \(\mathcal A_{{\rm cl},ij}\).
Then, condition \eqref{NONcondi1_3} implies that \(\mathcal O\) is a nonsingular \(M\)-matrix. By the same vector Lyapunov argument as in Theorem~\ref{lin_the}, the origin of the homogeneous interconnected closed-loop dynamics is exponentially stable. Consequently, in the presence of constant disturbances, the closed-loop trajectories remain bounded.

The remaining convergence argument follows the same KKT-based output-regulation reasoning as in Theorem~\ref{lin_the}. In particular, the local regulation errors satisfy \(e_i(t)\to0\), and the communication protocol yields asymptotically consistent estimates, i.e.,
$\hat y_{-i}^i(t)-y_{-i}(t)\to0$.
Together with the multiplier dynamics
$\dot\varsigma_i=R_i y_i-h_i$,
the limiting closed-loop steady state satisfies the KKT conditions of
the equality-constrained game. Since Assumption~\ref{ass_pseudo}
guarantees that the KKT system has the unique solution
\((y^\ast,\varsigma^\ast)\), it follows that
$y(t)\to y^\ast$.
Therefore, the decentralized control strategy \eqref{nonstate_control}
solves Problem~\ref{pro_2}.
\end{pf}

\section{Simulation Results}\label{Simulation}
To illustrate the proposed framework, we consider two case studies: (i) a network of unmanned aerial vehicles (UAVs) with linear dynamics subject to exogenous disturbances, and (ii) a formation of rotary-wing aerial vehicles with nonlinear dynamics subject to constant disturbances. In both cases, decentralized gains are synthesized entirely from noisy data via the SDPs derived in the theory and implemented under partial-decision information over connected communication graphs.

\begin{figure}[t]
  \centering
    \includegraphics[scale=0.4]{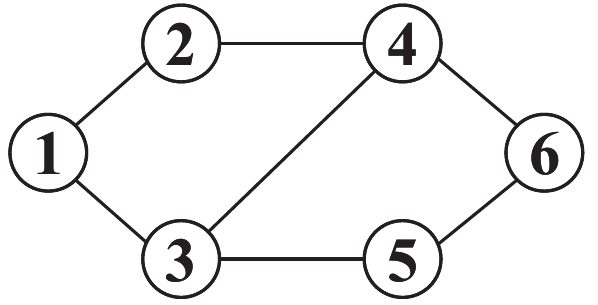}
    \caption{Communication graph of the UAV network.}\label{fig1}
  \centering
\end{figure}

\subsection{UAV Networks}
Consider a network of six UAVs communicating over a connected graph shown in Fig. \ref{fig1}. For each agent $i \in \mathcal{I}=\{1,\dots,6\}$, the dynamics are given by
\begin{align}
\dot{y}_i = v_i, \qquad
\dot{v}_i = -\frac{k_i}{m_i}v_i + u_i + d_i,
\end{align}
where $y_i,v_i,u_i\in\mathbb{R}^3$ denote the position, velocity, and control input, respectively. Defining $x_i=\operatorname{col}(y_i,v_i)$, the above dynamics can be written in the form of \eqref{linsys}. The agent parameters are
$m=[1.2, 1.5, 1.8, 1.35, 1.65, 1.4]\ \text{kg}$, and
$k=[0.24, 0.30, 0.40,$ $0.29, 0.27, 0.32]\ \text{N}\!\cdot\!\text{s}/\text{m}$.
The wind disturbance $d_i\in\mathbb{R}^3$ is generated by the exosystem
\begin{equation}
\dot{d}_i =\setlength{\arraycolsep}{2pt} \begin{bmatrix}
0 & \tfrac{\pi}{10} & -\tfrac{\pi}{20} \\
-\tfrac{\pi}{10} & 0 & \tfrac{\pi}{15}\\
\tfrac{\pi}{20}  & -\tfrac{\pi}{15}& 0
\end{bmatrix} d_i.    
\end{equation}
In this case study, each UAV is assigned an individual waypoint $r_i$, while the team is required to maintain cohesion at a prescribed altitude. This creates a trade-off between individual target tracking and formation preservation, and  motivates the following constrained game
\begin{subequations}
\begin{align}
\min_{y_i} \quad & a_{i}^{i} \| y_i - r_i \|^2 + \sum_{j \in \mathcal{I}, \, j \neq i} b_{i}^{j} \| y_i - y_j \|^2 ,\label{sin_cost_a}\\
\text{s.t.} ~~\quad & 
 [ 0 ~ 0 ~ 1] y_i = h_i.\label{sin_cost_b}
\end{align}
\end{subequations}
Here, the first term in \eqref{sin_cost_a} promotes tracking of the individual waypoint, while the second term promotes team cohesion. In \eqref{sin_cost_b}, $h_i=30$ specifies a common operating altitude of 30 meters. The weights are chosen as $a_i^i=3$ and $b_i^j=0.25$, under which the block strict diagonal dominance condition in Remark \ref{KKT-matrix} is satisfied, and hence Assumption \ref{ass_pseudo} holds.
The target positions are
$r_{1}=[-2,1,29]^{\top}$, $r_{2}=[1,-2,28]^{\top}$, $r_{3}=[2,-1,31]^{\top}$, $r_{4}=[-1,2,30]^{\top}$, $r_{5}=[-2,-2,32]^{\top}$, and $r_{6}=[0,1,30]^{\top}$,
with initial positions $(1,2,0)$, $(2,1,0)$, $(0,-1,0)$, $(2,0,0)$, $(1,-1,0)$, and $(1,0,0)$, respectively. The specified initial positions only determine where the UAVs start, whereas the target positions define the desired mission locations in the local cost functions. Together with the altitude constraints, they determine the constrained NE
$y^{*}=[-1.57~0.69~30~0.7~-1.56~30~1.34~-0.77~30~-0.68~1.5~30~-1.63~-1.71~30~-0.16~0.84~30]^{\top}$.

For the data-driven synthesis, offline data over $T=50$ sampling instants were collected. The chosen data length is consistent with the guideline in Remark \ref{rankre}, and the resulting data matrix satisfies the full-row-rank condition in Assumption \ref{rank}. During the data-collection phase, the control inputs $u_i$ were selected as excitation signals, and the corresponding local state trajectories in \eqref{eq:argue_no}  and outputs were recorded. The inputs $u_i$ were generated from the uniform distribution $[-0.5, 0.5]$. The initial conditions of the disturbances \(d_i(0)\) and the measurement noises \(\pi_i\) were sampled from \([-0.001,0.001]\). The disturbance trajectories \(d_i(t)\) were then generated by the exosystem above. The resulting noisy data were used to compute the decentralized control gains via the SDP \eqref{condi1}.

Under the synthesized strategy \eqref{eq:ct}, Fig. \ref{fig2} shows that the UAVs move from their initial positions (circles) to the constrained NE (crosses), whose altitude component satisfies the prescribed constraint $h_i=30$ m.
Fig. \ref{fig3} further shows that the regulated output errors $e_i$ converge asymptotically to zero, in agreement with the theoretical analysis.

\begin{figure}[htbp]
  \centering
  \includegraphics[scale=0.52] {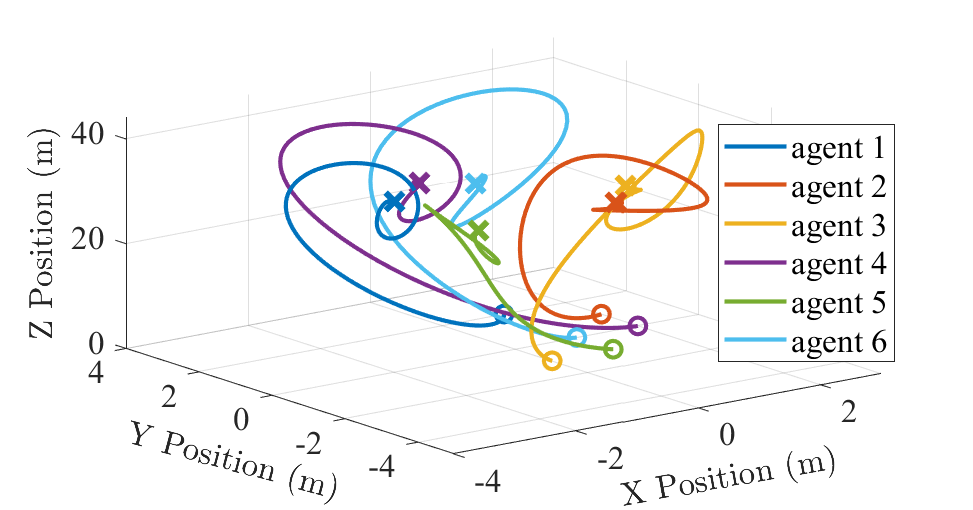}
  \caption{Trajectories of the UAV network.}\label{fig2}
\end{figure}

\begin{figure}[htbp]
  \centering
  \includegraphics[scale=0.48] {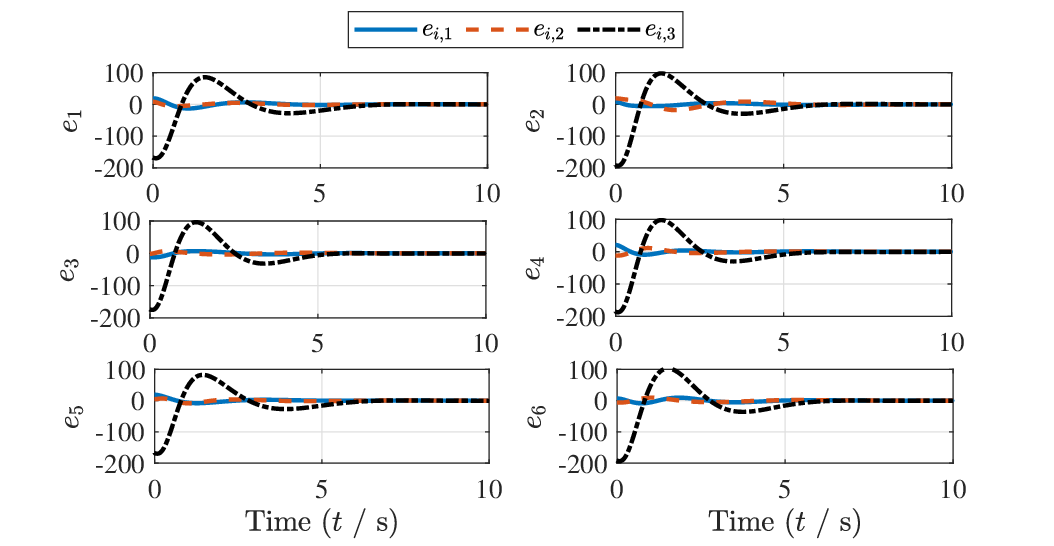}
  \caption{Regulation errors of the UAV network.}\label{fig3}
\end{figure}

\noindent\textit{1) Comparison with the model-based approach.}
For comparison, we implement a model-based controller in which the exact system matrices, available only in simulation, are used to compute the decentralized feedback gains. The NE model, communication protocol, and internal model are kept the same. As shown in Fig.~\ref{figq1_0}, both controllers drive the local regulation errors to zero, with nearly identical transient responses. This shows that the proposed data-driven design achieves performance comparable to the model-based one, although it is synthesized directly from noisy offline data without using the exact system matrices.

\begin{figure}[htbp]
  \centering
  \includegraphics[scale=0.48] {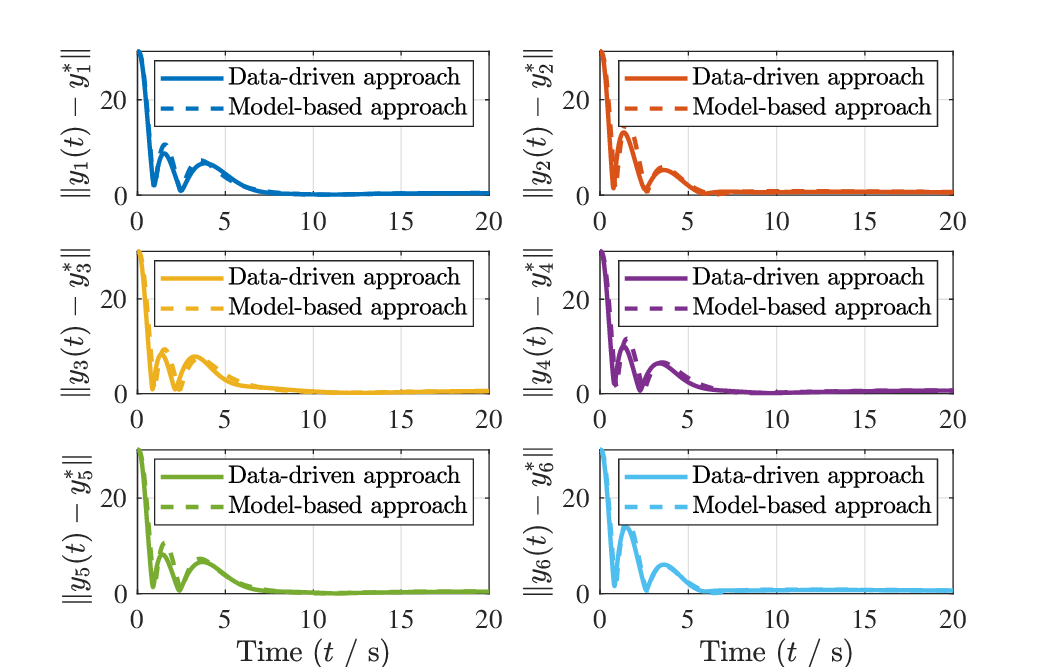}
  \caption{NE seeking error of each agent under data-driven control and model-based control.}\label{figq1_0}
\end{figure}

\noindent\textit{2) Comparison with two reduced designs.}
To show the role of the proposed modules, we compare the full controller with two reduced designs, namely, one without the communication protocol \eqref{eq:cpy} and one without the KKT dynamics \eqref{eq:NEMa}.
Fig.~\ref{figq1_1} shows that the proposed communication protocol drives the average estimation error to zero, whereas removing \eqref{eq:cpy} results in a persistent estimation mismatch. Fig.~\ref{figq1_2} shows that the KKT dynamics \eqref{eq:NEMa} are essential for asymptotic feasibility, since without them the equality-constraint residual $Ry-h$ does not converge to zero. Fig.~\ref{figq1_3} further reports the NE seeking error $\|y-y^*\|$. The proposed controller drives this error to zero, while both reduced designs fail to converge to the NE. 

\begin{figure}[htbp]
  \centering
  \includegraphics[scale=0.5] {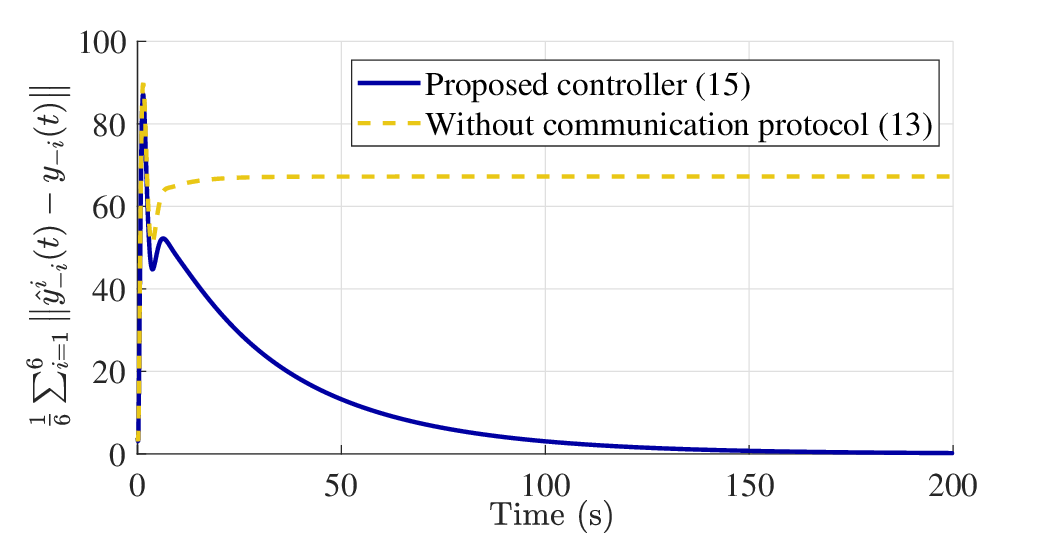}
  \caption{Average estimation error with and without the communication protocol (13).}\label{figq1_1}
\end{figure}

\begin{figure}[htbp]
  \centering
  \includegraphics[scale=0.5] {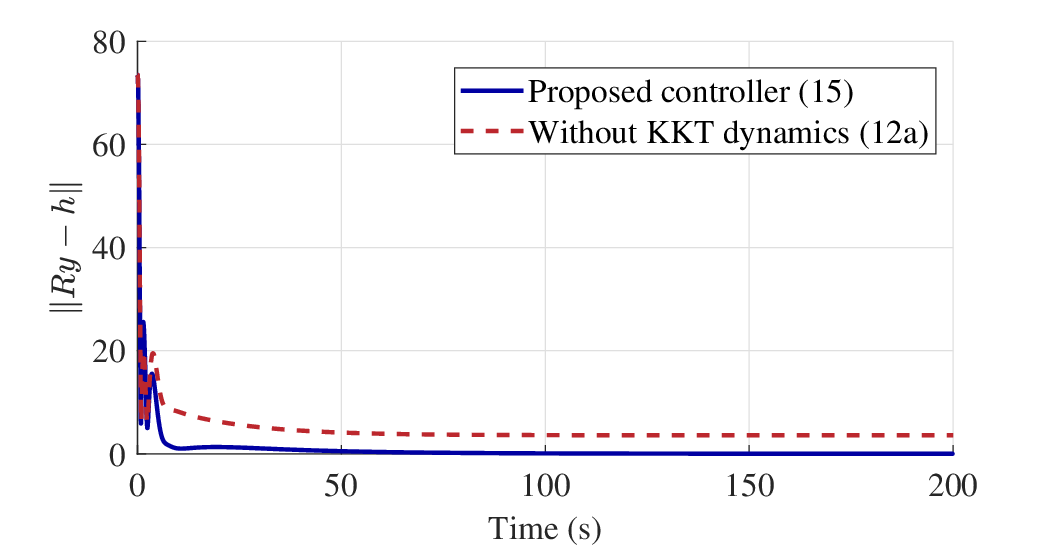}
  \caption{Equality-constraint residual with and without the  KKT dynamics (12a).}\label{figq1_2}
\end{figure}

\begin{figure}[htbp]
  \centering
  \includegraphics[scale=0.49] {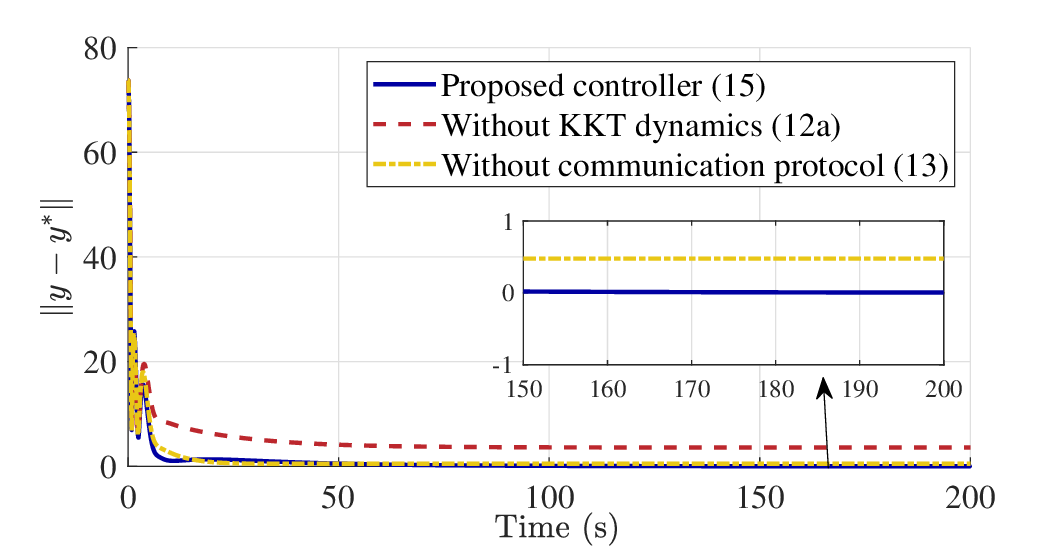}
  \caption{NE seeking error under the proposed and reduced designs.}\label{figq1_3}
\end{figure}

\begin{figure}[t]
  \centering
  \includegraphics[scale=0.7] {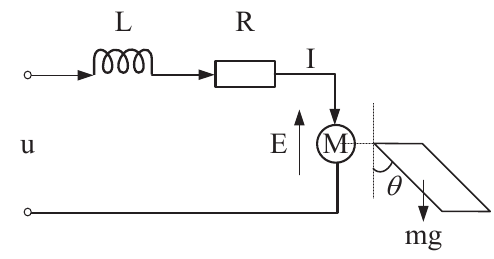}
  \caption{Schematic diagram of the wing load.}\label{fig4}
  \centering
\end{figure}

\subsection{Rotary-Wing Aerial Vehicle Formation}
To evaluate the proposed nonlinear design on heterogeneous agents, we consider a network of five rotary-wing aerial vehicles communicating over the connected graph in Fig. \ref{fig5}. Following the rotary-wing model in \citet{ZhaoAppointed2022}, but allowing agent-dependent normalized coefficients, the dynamics of agent $i\in\mathcal I=\{1,\dots,5\}$ are given by
\begin{equation}
\dot x_i = A_i x_i + A_{1i}\sin(\theta_i) + B_i u_i + D_i w_i,~
y_i=C_i x_i,
\end{equation}
where $x_i=[\theta_i,\dot\theta_i,I_i]^\top$, $y_i=[\theta_i,I_i]^\top$, and $\theta_i$, $I_i$, and $u_i$ denote the angular displacement, motor current, and control input, respectively. Equivalently, the model can be written in the form of \eqref{nonlinsys} with
\begin{equation*}
\mathcal Z(x_i)=\operatorname{col}(x_i,\sin(\theta_i)),\qquad
\mathcal A_i=[\,A_i\;\;A_{1i}\,].
\end{equation*}
The system matrices are
\[
A_i=
\begin{bmatrix}
0 & 1 & 0\\
-k_{1i} & 0 & k_{2i}\\
0 & -k_{3i} & -k_{4i}
\end{bmatrix},\qquad
A_{1i}=
\begin{bmatrix}
0\\
-k_{5i}\\
0
\end{bmatrix},
\]
\[
B_i=
\begin{bmatrix}
0 & 0\\
b_{1i} & 0\\
0 & b_{2i}
\end{bmatrix},\qquad
D_i=d_i I_3,\qquad
C_i=
\begin{bmatrix}
1 & 0 & 0\\
0 & 0 & 1
\end{bmatrix}.
\]
The heterogeneous coefficients are chosen as
\begin{align*}
k_1&=[1.00,\,0.92,\,1.10,\,0.95,\,1.05],\\
k_2&=[1.00,\,1.05,\,0.95,\,1.08,\,0.92],\\
k_3&=[0.50,\,0.46,\,0.54,\,0.48,\,0.52],\\
k_4&=[0.20,\,0.24,\,0.18,\,0.22,\,0.26],\\
k_5&=[0.40,\,0.36,\,0.44,\,0.38,\,0.42],\\
b_1&=[1.00,\,0.90,\,1.10,\,0.95,\,1.05],\\
b_2&=[1.00,\,1.08,\,0.92,\,1.04,\,0.96],\\
d&=[1.00,\,1.05,\,0.95,\,1.08,\,0.92].
\end{align*}
The initial output values are chosen as $[0.1,0]^{\top}$, $[-0.1,0]^{\top}$, $[0,0]^{\top}$, $[0.2,0]^{\top}$, and $[-0.2,0]^{\top}$. The constant exogenous signal is chosen as $w_i=[0.2,\,0.1,\,0.1]^\top$ for all $i\in\mathcal I$.

In this case study, each rotary-wing aerial vehicle is assigned an individual target output $r_i$ while coordinating with the others under a prescribed motor-current requirement. This leads to the following constrained game
\begin{subequations}
\begin{align}
\min_{y_i} \quad & a_{i}^{i} \| y_i - r_i \|^2 + \sum_{j \in \mathcal{I}, \, j \neq i} b_{i}^{j} \| y_i - y_j \|^2, \label{nsin_cost_a}\\
\text{s.t.} ~~\quad & [ 0 ~ 1] y_i = h_i. \label{nsin_cost_b}
\end{align}
\end{subequations}
Here, the cost in \eqref{nsin_cost_a} combines the individual regulation objective and the coordination requirement. Constraint \eqref{nsin_cost_b} imposes the prescribed motor-current level, with $h_i = 2$ A. The weights are chosen as $a_i^i=4$ and $b_i^j=0.3$, under which the block strict diagonal dominance condition in Remark \ref{KKT-matrix} is satisfied, and hence Assumption \ref{ass_pseudo} holds. 
The target outputs are chosen as $[0.5,1.8]^{\top}$, $[0.4,1.9]^{\top}$, $[0.6,1.8]^{\top}$, $[0.3,1.7]^{\top}$, and $[0.5,1.9]^{\top}$. The corresponding constrained NE is $y^\ast= [0.4992,2,0.4004,2,0.5873,2,0.3070,2,0.5061,2]^\top$.
For the data-driven synthesis, offline input-state data over $T=30$ sampling instants were collected. The inputs $u_i$ were generated from the uniform distribution $[-0.5,0.5]$, and the measurement noises during data collection were generated from $[-0.001,0.001]$. The resulting noisy dataset was then used to solve the LMIs in \eqref{NONSDP} and construct the decentralized controller in \eqref{nonstate_control}.

As illustrated in Fig. \ref{fig6}, the proposed strategy \eqref{nonstate_control} drives the outputs to the NE. Moreover, Fig. \ref{fig7} shows that the corresponding regulation errors converge asymptotically to zero, which is consistent with the theoretical stability result and illustrates the effectiveness of the proposed data-driven approach.

\begin{figure}[htbp]
  \centering
    \includegraphics[scale=0.4]{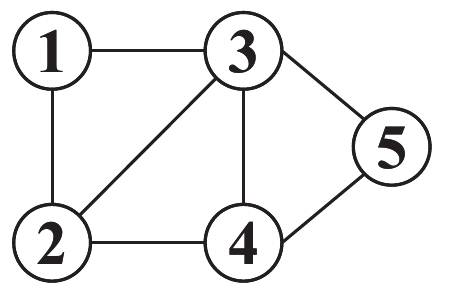}
    \caption{Communication graph of rotary-wing aerial vehicles.}\label{fig5}
  \centering
\end{figure}

\begin{figure}[htpb]
  \centering
  \includegraphics[scale=0.48] {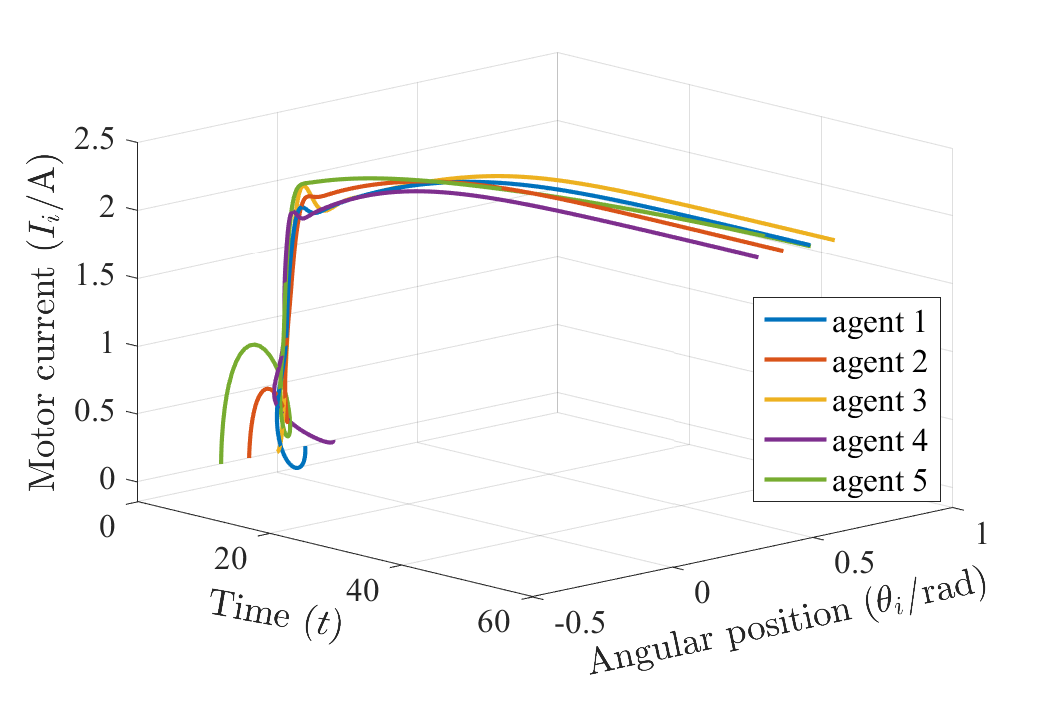}
  \caption{Trajectories of rotary-wing aerial vehicles.}\label{fig6}
\end{figure}

\begin{figure}[htpb]
  \centering
  \includegraphics[scale=0.5] {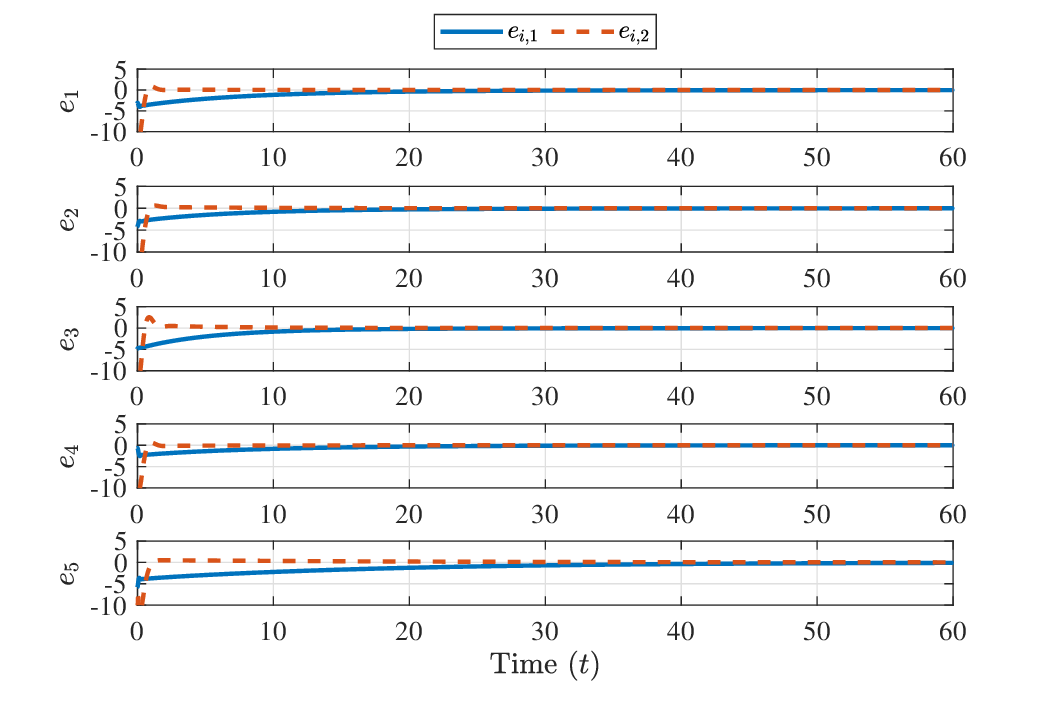}
  \caption{Regulation errors of rotary-wing aerial vehicles.}\label{fig7}
\end{figure}

\section{Conclusions}
This paper developed a unified data-driven framework for decentralized NE seeking in multi-agent systems with equality constraints and partial-decision information, under unknown dynamics and exogenous disturbances. By recasting the problem as a cooperative output regulation problem, the proposed approach integrates an NE model, a communication protocol, an internal model, and a direct data-driven controller, without explicit system identification.
For linear systems subject to exogenous disturbances generated by known exosystems, data-based SDPs were derived to guarantee closed-loop stability and asymptotic convergence to the NE. The framework was further extended to a structured class of nonlinear systems with constant disturbances by combining integral control and quadratic constraints. Numerical simulations illustrated the effectiveness of the proposed approach.
Future work will consider inequality constraints, time-varying dynamics, and unknown exosystem matrices to broaden the applicability of the proposed framework.

\bibliographystyle{plainnat}
\bibliography{ref}
\end{document}